\documentclass[amsthm]{elsart}

\usepackage{graphicx}
\usepackage{indentfirst}
\usepackage{url}
\usepackage{color}
\usepackage{yjsco}
\usepackage{natbib}
\usepackage{amsmath}
\usepackage{amsfonts}
\usepackage{amsthm}
\usepackage{amssymb}
\usepackage{bm}
\usepackage{enumitem}

\usepackage{comment}
\usepackage{algpseudocode}

\graphicspath{{img/}}
\def \coeffs{{\rm coeffs}}
\def \codim {{\rm codim}}
\def \OWD  {{\rm OWD}}
\def \Res  {{\rm Res}}
\def \MCproj  {{\rm MCproj}}

\def \discrim  {{\rm discrim}}
\def \sqrfree  {{\rm sqrfree}}

\def \lc  {{\rm lc}}
\def \Bproj  {{\tt Bp}}
\def  \zero {{\rm Zero}}
\def  \Nproj {{\tt Np}}

\def  \Hproj {{\tt Hp}}

\def \Proineq {{\tt DPS}}

\def \Bprojection {{\tt Bproj}}

\def \RR {{\mathbb R}}
\def \ZZ {{\mathbb Z}}

\newcommand{\va}{\bm{\alpha}}
\newcommand{\vb}{\bm{\beta}}
\newcommand{\xx}{\bm{x}}
\newcommand{\yy}{\bm{y}}
\newcommand{\zz}{\bm{z}}

\def \RAGlib {{\tt RAGlib}}
\def \FI {{\tt FI}}
\def \QEPCAD {{\tt QEPCAD}}
\def \PCAD {{\tt PCAD}}
\def \TwoPro {{\tt PSD-HpTwo}}
\def \TwoHp {{\tt HpTwo}}

\def \FI {{\tt FI}}
\def \QEPCAD {{\tt QEPCAD}}
\def \PCAD {{\tt PCAD}}
\def \TCPT {{\tt CMT}}

\def \OO {{\mathcal{O}}}

\newcommand\numberthis{\addtocounter{equation}{1}\tag{\theequation}}

\newtheorem{ex}{Example}   

\begin{document}

\begin{frontmatter}

\title{\textbf{Open Weak CAD and Its Applications}}

\thanks{This research was partly supported by US National Science Foundation Grant 1319632, National Science Foundation of China Grants 11290141, 11271034 and 61532019, and China Scholarship Council.}

\author{Jingjun Han, 
Liyun Dai}
\ead{hanjingjunfdfz@gmail.com, dailiyun@pku.edu.cn}
\address{LMAM $\&$ School of Mathematical Sciences\\ $\&$ Beijing International Center for Mathematical Research\\Peking University, Beijing 100871, China}

\author{Hoon Hong}
\address{Department of Mathematics, North Carolina State University, Raleigh NC 27695, USA}
\ead{hong@ncsu.edu}
\ead[url]{www.math.ncsu.edu/\~{}hong}

\author{Bican Xia}
\address{LMAM $\&$ School of Mathematical Sciences, Peking University, Beijing 100871, China}
\ead{xbc@math.pku.edu.cn}
\ead[url]{www.math.pku.edu.cn/is/\~{}xbc/}

\begin{abstract}
The concept of open weak CAD is introduced. Every open CAD is an open weak CAD. On the contrary, an open weak CAD is not necessarily an open CAD. An algorithm for computing projection polynomials of open weak CADs is proposed. The key idea is to compute the intersection of projection factor sets produced by different projection orders. The resulting open weak CAD often has smaller number of sample points than open CADs.

The algorithm can be used for computing sample points for all open connected components of $ f\neq0$ for a given polynomial $f$. It can also be used for many other applications, such as testing semi-definiteness of polynomials and copositive problems. In fact, we solved several difficult semi-definiteness problems   efficiently by using the algorithm. Furthermore, applying the algorithm to copositive problems, we find an explicit expression of the polynomials producing open weak CADs under some conditions, which significantly improves the efficiency of solving copositive problems.


\end{abstract}

\end{frontmatter}

\section{Introduction}
\label{secc:intro}
In this paper, we introduce the concept of {\em\ open weak CAD}, provide an algorithm for computing corresponding projection polynomials, and illustrate its usefulness through various applications. In the following, we elaborate on the above sentence.

Cylindrical Algebraic Decomposition (CAD)\ is a fundamental concept and tool for computational real algebraic geometry with numerous applications. It was introduced by
\citet{collins1,Caviness1998} and have been improved
by  many \citep{McCallum, Hong:90a, Collins_Hong:91, Hong:92a, LS93, Ren, BPR96,  McCallumeq, AnaiW01, Brown01a, Strzebonski06, Hong_Safey2012,Chen_Maza2014}.

For polynomials in $n$ variables, a CAD is a finite collection of sign-invariant cells satisfying the following two requirements : (a) the cells constitute a  decomposition  of the whole $n$-dimensional space.  (b) the cells  are cylindrically arranged.  Such a decomposition is typically constructed in two stages: (1)\  compute   a kind of triangular set of  polynomials through repeated projection, resulting  in so-called projection polynomials (2)\ carrying back substitutions by repeatedly solving the projection polynomials.

It  was immediately observed that the computation   of CAD can be     time-consuming. Hence there have been intensive and diverse effort to improve its computational efficiency.  For instance, it was observed that the algorithms spend  huge amount of time on  constructing low dimensional cells and that those cells are often not needed in various applications. Hence a relaxed notion called {\em\ open CAD} or {\em\ generic CAD} was introduced, where  low dimensional cells are ignored. Consequently, it relaxes the requirement (a) that the cells constitute a decomposition of the {\em whole\/} space.

In this paper, we introduce a further relaxed notion called {\em open weak CAD}, where we  also relax  the requirement (b)  that the cells  are cylindrically arranged.   Technically, the cylindricity is intimately related to  delineability.   We replace the original delineability requirement with a weaker version in such a way that it still captures sufficient amount of geometric information needed in various applications.

Furthermore, we  provide an algorithm for computing projection polynomials of open weak CADs. The key idea is to compute the intersection of projection factor sets produced by different projection orders. The resulting open weak CAD often has smaller number of cells than open CADs.

We illustrate the usefulness of open weak CAD theory and algorithm by tackling several application problems. First we show how to compute sample points for all open connected components of $ f\ne 0$ for a given polynomial $f$. Next we show how to test polynomial inequalities. Finally we  show how to tackle  co-positiveness problems; we find an explicit expression of the polynomials producing open weak CADs under some conditions, which significantly improves the efficiency of solving copositive problems.

The structure of this paper is as follows. In Section \ref{sec:problem}, we introduce a notion of  open weak delineability  and open weak CAD and then state the problem of finding projection polynomials. In Section \ref{sec:pre}, we review basic definitions, lemmas and concepts of CAD. In Section \ref{sec:refined}, we introduce several properties of open weak delineable, provide an algorithm for computing projection polynomials of open weak CADs  and prove its correctness. In Section \ref{sec:reduced_open_cad}, we apply open weak CAD theory and algorithm to  compute open sample. In Section \ref{sec:improved}, we apply open weak CAD theory and algorithm and some previous work \citep{han2016proving} to prove  polynomial inequalities.    In Section \ref{sec:copositive}, we again  apply open weak CAD theory and algorithm to  test co-positiveness. Section \ref{sec:applicat}, we provide  several examples which demonstrate the effectiveness of the algorithms. 


 \section{Problem: Open weak CAD}\label{sec:problem}


 In this section, we give a precise statement of the problem.  We begin by  introducing several notions.

 \begin{defn}[Open weak delineable] \label{def:weakopendeli}
Let  $f\in \RR[x_1,\ldots,x_n]$ and let $S$ be an open set of $\RR^j (1\le j< n)$. We say that $f$ is  {\em open weak delineable} (OWD) on $S$ if, for any open connected component $U\subseteq \RR^n$ defined by $f\neq0$, we have either
$$S\subseteq \pi_j^n(U)  \,\,\,\text{or}\,\,\, S\cap \pi_j^n(U)=\emptyset,$$
where $\pi_j^n(x_1,\ldots,x_n)=(x_1,\ldots,x_j)$.
Let $h\in \RR[x_1,\ldots,x_j]$ for $j\le n.$ We say
that $f$ is {\em open weak delineable\/} over~$h$ in $\RR^j$, if $f$ is open weak delineable on any open connected component of $h\neq0$ in $\RR^j$.
\end{defn}

\begin{ex}\label{ex:wod}
Let $$f=(x_2^2-1)^2-x_1\in \RR[x_1,x_2].$$
The plot of $f=0$ is given by
\begin{figure}[ht]
\begin{centering}
        \includegraphics[width=2.0in]{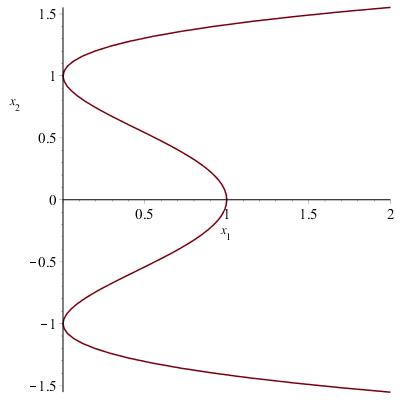}
\caption{Example \ref{ex:wod}\label{fig:ex21} }
\end{centering}
\end{figure}

\noindent
Note that $f$ is analytically delineable (\citet{collins1}, \citet{McCallum2}) and also open weak delineable on the set $(-\infty, 0)$.
Note that $f$ is {\em not\/} analytically delineable  but {\em is\/} open weak delineable on the set $(0,\infty)$. Note also that $f$ is open weak delineable over~$h=x_1$ in $\RR.$
\end{ex}



\begin{defn}[Open weak CAD]
Let $f\in \RR[x_1,\ldots,x_n]$. A decomposition of $\RR^j$, $1\leq j< n$, is called an {\em open weak CAD\/} of $f$ in $\RR^j$ if and only if $f$ is open weak delineable on every $j$-dimensional open set in the decomposition. 
\end{defn}

\begin{ex}
Let $f$ be the  polynomial from Example~\ref{ex:wod}, 
\[\{(-\infty,0),[0,0],(0,\infty)\}\]
is an open weak CAD of $f$ in $\RR$.
\end{ex}

Finally, we are ready to state the problem precisely.\medskip

\noindent \textbf{Problem}. (Projection polynomials of open weak CAD) Devise an algorithm with the following specification.
\medskip
\begin{description}[leftmargin=3em,style=nextline,itemsep=0.5em]
\item[\sf In:]   $f\in \ZZ[x_1,\ldots,x_n]$
\item[\sf Out:]  $h_1, \ldots, h_{n-1}$ where $h_j\in \ZZ[x_1,\ldots,x_j]$
such that $f$ is open weak delineable over $h_j$ in $\RR^j$.
\end{description}

\begin{ex}\label{ex:1}
Consider the following polynomial. \medskip
\begin{description}[leftmargin=3em,style=nextline,itemsep=0.5em]
\item[\sf In:]   $f=(x_3^2+x_2^2+x_1^2-1)(4x_3+3x_2+2x_1-1)\in \ZZ[x_1, x_2,x_3]$
\item[\sf Out:]  $h_1 =(x_1-1)(x_1+1)(29x_1^2-4x_1-24)((20x_1^2-4x_1-15)^2+(13x_1^2-4x_1-8)^2)\in \ZZ[x_1]$
\item[]             $h_2 =(x_2^2+x_1^2-1)(25x_2^2+12x_2x_1+20x_1^2-6x_2-4x_1-15)\in \ZZ[x_1, x_2] $
\end{description} \medskip
The left plot in Figure~\ref{fig:ex1ab} below shows the open weak CAD of $f$  produced by $h_1$ and~$h_2$.
The factor $(20x_1^2-4x_1-15)^2+(13x_1^2-4x_1-8)^2$ in $h_1$ does not have real root and thus it does not contribute to the open weak CAD.
\end{ex}
\begin{figure}[h]
\begin{center}
\includegraphics[width=2.5in]{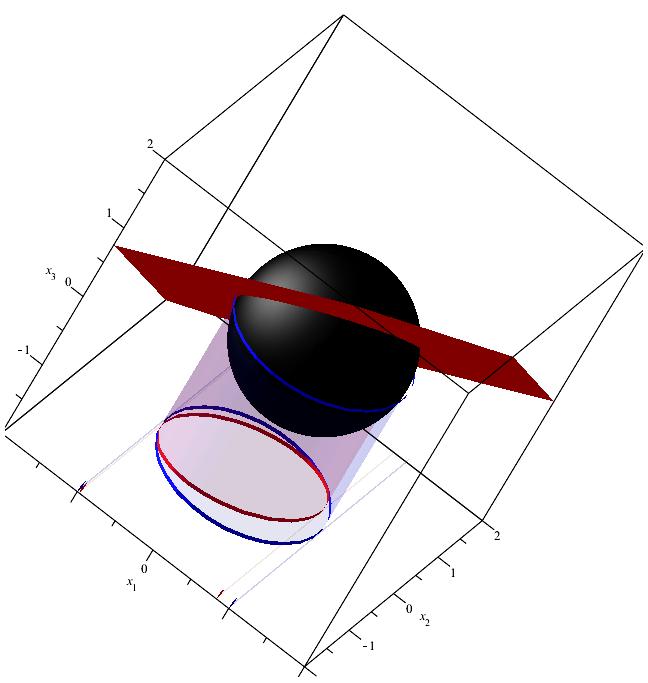}
\includegraphics[width=2.5in]{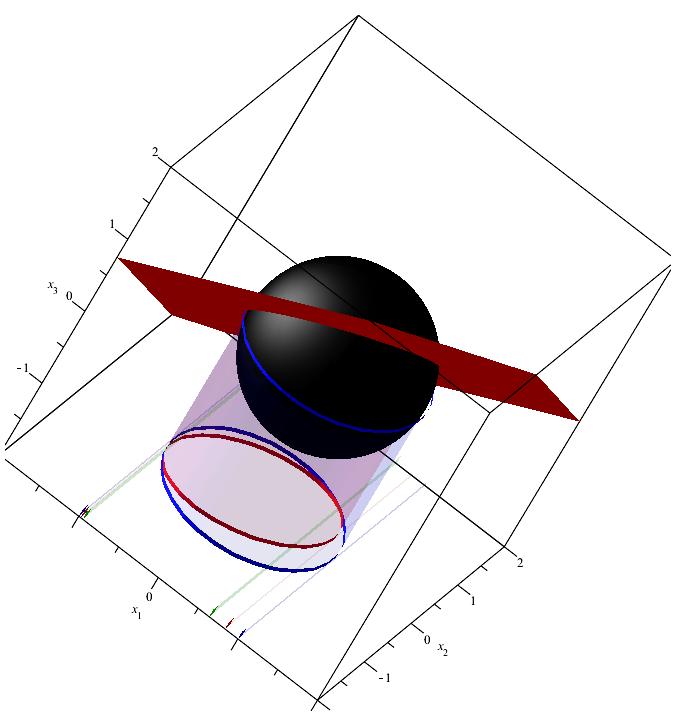}
\end{center}
\caption{\label{fig:ex1ab}Example \ref{ex:1} }
\end{figure}

\begin{rem} For comparison, if we apply an open CAD  algorithm on the above $f$, one would obtain the following
output\medskip
\begin{description}[leftmargin=3em,style=nextline,itemsep=0.5em]
\item[\sf Out:]  $h_1=(x_1-1)(x_1+1)(29x_1^2-4x_1-24)(13x_1^2-4x_1-8)\in \ZZ[x_1]$
\item[]          $h_2 =(x_2^2+x_1^2-1)(25x_2^2+12x_2x_1+20x_1^2-6x_2-4x_1-15)\in \ZZ[x_1, x_2]$
\end{description}  \medskip
The right plot in Figure~\ref{fig:ex1ab} shows the open  CAD of $f$  produced by $h_1$ and~$h_2$.
Note that it has more cells than the open weak CAD (on the left).
\end{rem}

\begin{rem}
It is natural to wonder whether the multivariate discriminants of $f$ always produce open weak CADs. Unfortunately, this is not true since the discriminant ${\rm discrim}(f,[x_n,\ldots,x_{j+1}])$ (the multivariate discriminant of $f$ with respect to $x_n,\ldots,x_{j+1}$) may vanish identically and thus does not always produces an open weak CAD of $\RR^j$.
One may also wonder whether if the multivariate discriminants of $f$ do produce open weak CADs, then they would be the smallest open weak CADs. Unfortunately this is not true  either. In Example \ref{ex:wod}, it has been shown that $x_1$ produces an open weak CAD of $\RR$ with $2$ open intervals. But the discriminant ${\rm discrim}(f,x_2)=-256(x_1-1)x_1^2$ produces an open weak CAD of $\RR$ with $3$ open intervals.
\end{rem}

\begin{rem}
The output of the above problem is a list of ``projection'' polynomials, not an open weak CAD. As usual, we can compute sample points in an open weak CAD of $f$ from the projection polynomial $h_j$ by standard lifting technique. We say that $h_j$ produce an open weak CAD. Thus, in this paper, we will focus ourselves on the problem of finding projection polynomial of Open Weak CAD.
\end{rem}




\section{Preliminaries}
\label{sec:pre}
If not specified, for a positive integer $n$, let $\xx_n$ be the list of variable $(x_1,\dots,x_n)$ and $\va_n$ and $\vb_n$ denote the point $(\alpha_1,\dots,\alpha_n)\in \RR^n$ and $(\beta_1,\dots,\beta_n)\in \RR^n$, respectively.



\begin{defn}
Let $f\in \ZZ[\xx_n]$, denote by $\lc(f,x_i)$ and $\discrim(f,x_i)$ the {\em leading coefficient} and the {\em discriminant} of $f$ with respect to (w.r.t.) $x_i$, respectively. 
The set of real zeros of $f$ is denoted by $\zero(f)$. Denote by $\zero(L)$ or $\zero(f_1,\ldots,f_m)$ the common real zeros of $L=\{f_1,\ldots,f_m\}\subset \ZZ[\xx_n].$
The {\em level} for $f$ is the biggest $j$ such that $\deg(f,{x_j})>0$ where $\deg(f,x_j)$ is the degree of $f$ w.r.t. $x_j$.
        For polynomial set $L\subseteq \ZZ[\xx_n]$, $L^{[i]}$ is the set of polynomials in $L$ with level $i$.
\end{defn} 

\begin{defn}\label{de:sqrfree}
If $h$$\in \ZZ[\xx_n]$ can be factorized in $\ZZ[\xx_n]$ as:
        $$h=al_{1}^{2j_1-1}\cdots l_t^{2j_t-1}{h_1}^{2i_1}\cdots {h_m}^{2i_m},$$
        where $a\in \ZZ$, $t\ge 0, m\ge 0$, $l_i$$(i=1,\ldots,t)$ and $h_i$$(i=1,\ldots,m)$ are
pairwise different irreducible primitive polynomials with positive leading coefficients $($under a suitable ordering$)$ and positive degrees in $\ZZ[\xx_n]$,
then define
        \begin{align*}
        &\sqrfree(h)=l_{1}\cdots l_t{h_1}\cdots {h_m},\\ 
                &\sqrfree_1(h)=\{l_i,i=1,2,\ldots,t\},\\
                &\sqrfree_2(h)=\{h_i,i=1,2,\ldots,m\}.
        \end{align*}
        If $h$ is a constant, let $\sqrfree(h)=1,$ $\sqrfree_1(h)=\sqrfree_2(h)=\{1\}.$
\end{defn}

In the following, we introduce some basic known concepts and results of CAD. The reader is referred to \cite{collins1,Hong:90a,Collins_Hong:91,McCallum1,McCallum2,Brown01a} for a detailed discussion on the properties of CAD.

\begin{defn}\citep{collins1,McCallum1}\label{def:delineable}
        An $n$-variate polynomial $f(\xx_{n-1},x_{n})$ over the reals is said to be {\em delineable} on a subset $S$ (usually connected) of $\RR^{n-1}$ if
        (1) the portion of the real variety of $f$ that lies in the cylinder $S\times \RR$ over $S$ consists of the union of the graphs of some $t\ge0$ continuous functions $\theta_1<\cdots<\theta_t$ from $S$ to $\RR$; and
        (2) there exist integers $m_1,\ldots,m_t\ge1$ s.t. for every $a\in S$, the multiplicity of the root $\theta_i(a)$ of $f(a,x_n)$ (considered as a polynomial in $x_n$ alone) is $m_i$.
\end{defn}

\begin{defn}\citep{collins1,McCallum1}
        In the above definition, the $\theta_i$ are called the real root functions of $f$ on $S$, the graphs of the $\theta_i$ are called the $f$-{\em sections} over $S$, and the regions between successive $f$-sections are called $f$-{\em sectors}.

\end{defn}

\begin{thm}\label{thm:McCallum}\citep{McCallum1,McCallum2}
        Let $f(\xx_n,x_{n+1})$ be a polynomial in $\ZZ[\xx_n,x_{n+1}]$ of positive degree and $\discrim(f,x_{n+1})$ is a nonzero polynomial. Let $S$ be a connected submanifold of $\RR^n$ on which $f$ is degree-invariant and does not vanish identically, and in which $\discrim(f,x_{n+1})$ is order-invariant. Then $f$ is analytic delineable on $S$ and is order-invariant in each $f$-section over $S$.
\end{thm}
Based on this theorem, McCallum proposed the projection operator \MCproj, which consists of the discriminant of $f$ and all coefficients of $f$.

\begin{thm}\label{thm:Brown}\citep{Brown01a}
        Let $f(\xx_n,x_{n+1})$ be a $(n+1)$-variate polynomial of positive degree 
in $x_{n+1}$ such that $\discrim(f,x_{n+1})$ $\neq0$. Let $S$ be a connected submanifold of $\RR^n$ in which $\discrim(f,x_{n+1})$ is order-invariant, the leading coefficient of $f$ is sign-invariant, and such that $f$ vanishes identically at no point in $S$. $f$ is degree-invariant on $S$.
\end{thm}
Based on this theorem, Brown obtained a reduced McCallum projection in which only leading coefficients, discriminants and resultants appear. The Brown projection operator is defined as follows.
\begin{defn} \label{def:brown projection}\citep{Brown01a}
        Given a polynomial $f\in \ZZ[\xx_n]$, if $f$ is with level $n$,
        the Brown projection operator for $f$ is
        $$\Bproj(f,[x_{n}])=\Res(\sqrfree(f),\frac{\partial (\sqrfree(f))}{\partial x_{n}}, x_{n}).$$
        Otherwise, $\Bproj(f,[x_{n}])=f$.
        If $L$ is a polynomial set with level $n$, then
        \begin{align*}
                \Bproj(L,[x_{n}])=&\bigcup_{f\in L}\{\Res(\sqrfree(f),\frac{\partial (\sqrfree(f))}{\partial x_{n}}, x_{n})\}\bigcup\\
                &\bigcup_{f,g\in L, f\neq g}\{\Res(\sqrfree(f),\sqrfree(g),x_{n})\}.
        \end{align*}
        Define
        \begin{align*}
                &\Bproj(f,[x_{n},x_{n-1},\ldots, x_i])\\
                =&\Bproj(\Bproj(f,[x_{n},x_{n-1},\ldots,x_{i+1}]),[x_i]).
        \end{align*}
\end{defn}

The following definition of {\em open CAD\/}  is essentially the GCAD introduced in \cite{Strzebonski}. For convenience, we use the terminology of open CAD in this paper.

\begin{defn} $($Open CAD$)$ \label{def:opencad}
        For a polynomial $f(\xx_n)\in \ZZ[\xx_n]$, an open CAD defined by $f(\xx_n)$ under the order $x_1\prec x_2\prec\cdots\prec x_n$ is a set of sample points in $\RR^n$ obtained through the following three phases:
\begin{enumerate}
       \item[(1)] [Projection] Use the Brown projection operator on $f(\xx_n)$,
                  let $$F=\{f,\Bproj(f,[x_n]),\ldots, \Bproj(f,[x_n,\dots,x_2])\};$$ 
       \item[(2)] [Base] Choose one rational point in each of the open intervals defined by the real roots of $F^{[1]}$; 
       \item[(3)] [Lifting] Substitute each sample point of $\RR^{i-1}$ for $\xx_{i-1}$ in $F^{[i]}$ to get a univariate polynomial $F_{i}(x_i)$ and then, by the same method as Base phase, choose sample points for $F_{i}(x_i)$. Repeat the process for $i$ from $2$ to $n$.
\end{enumerate}
\medskip
Sometimes, we say that the polynomial set $F$ produces the open CAD or simply, $F$ is an open CAD.
\end{defn}

\section{Open Weak CAD: Properties and Algorithm}\label{sec:refined}
In this section, we derive some basic properties of Open weak CAD, describe an algorithm (Algorithm \ref{alg:openweakcad}) for computing open weak CAD
and prove its correctness (Theorem~\ref{thm:openweakcad}). 

We first prove two simple but useful properties of open weak delineable. The first one is a transitive property.
\begin{prop}  \label{prop:weaktransitive}
 Let $f_n(\xx_n)\in\ZZ[x_1,\ldots,x_n]$, and $S$ be an open set of $\RR^j (1\le j< n)$. Suppose that there exists $k (j\le k\le n)$ and a polynomial $f_k(\xx_k)\in\ZZ[x_1,\ldots,x_k]$ such that $f_k(\xx_k)$ is open weak delineable on $S$, and $f_n(\xx_n)$ is open weak delineable on $f_k(\xx_k)$. Then $f_n(\xx_n)$ is open weak delineable on $S$.
\end{prop}
\begin{proof}
Let $U\subseteq \RR^n$ be any open connected component defined by $f_n\neq0$ such that $\pi_j^n(U)\cap S\neq\emptyset$. We have $\pi_j^{k^{-1}}(S)\cap\pi_k^n(U)\neq\emptyset$ since $\pi_j^n(U)=\pi_k^n(\pi_j^k(U))$. Let $S'\subseteq \RR^k$ be any open connected component defined by $f_k\neq0$ such that $S'\cap \pi_j^{k^{-1}}(S)\cap\pi_k^n(U)\neq\emptyset$. Now, $S'\subseteq \pi_k^n(U)$ and $S\subseteq\pi_{j}^k(S')$ since $f_n(\xx_n)$ is open weak delineable on $f_k(\xx_k)$, and~$f_k(\xx_k)$ is open weak delineable on $S$. Hence, $\pi_j^n(U)\supseteq \pi_{j}^k(S')\supseteq S$.
\end{proof}
Before stating the next property, let us take an example to illustrate our motivation. Let $f(x,y,z)=x^2+y^2\in\ZZ[x,y,z]$. In this case $f\neq0$ has only one open component $U=\{(x,y)\in\RR^2|x^2+y^2\neq0\}$. Let $S_1=\{(x,y)\in\RR^2|x>0\}$, $S_2=\{(x,y)\in\RR^2|x<0\}$, $S_3=\{(x,y)\in\RR^2|y>0\}$, $S_4=\{(x,y)\in\RR^2|y<0\}$. It is clear that $S_i\subseteq \pi_2^3 (U)$ $(i=1,\ldots,4)$, and $f$ is OWD over $x$ ($y$) in $\RR^2$, respectively. Since $\cup_{i=1}^4 S_i\subseteq \pi_2^3 (U)$ and $\cup_{i=1}^4 S_i=(\RR^2\backslash\zero(x))\cup(\RR^2\backslash\zero(y))=\RR^2\backslash\zero(x,y)=\RR^2\backslash\zero(x^2+y^2)$,
 $f$ is OWD over $x^2+y^2$. We note that $f$ is not OWD over $\gcd(x,y)=1$, while $\zero(\gcd(x,y))$ and $\zero(x^2+y^2)$ only differ at a closed set of codimension $2$. In general, we have the following Proposition.
\begin{prop} \label{prop:weakgcd}
 Let $f(\xx_n)\in\ZZ[\xx_n]$, suppose $f$ is OWD over $p_1,\ldots,p_t$ in $\RR^j$, $f$ is OWD over $p'=\sum_{i=1}^t p_i^2$ in $\RR^j$.
\end{prop}

\begin{proof}
Let $S\subseteq\RR^j$ be any open connected component defined by $p'\neq0$, $U\subseteq\RR^n$ be any open connected set defined by $f\neq0$. Let $S_1=\{\va|\va\in S,(\va\times\RR^{n-j})\cap U\neq\emptyset\}$, $S_2=\{\va|\va\in S,(\va\times\RR^{n-j})\cap U=\emptyset\}$. It is clear that $S=S_1\cup S_2$, $S_1\cap S_2=\emptyset$. For any $\va\in S$, we assume that $p_1(\va)\neq0$. There exists an open set $S_{\va}\subseteq \RR^j$ containing $\va$, such that $p_1(S_{\va})\neq0$. Since $f$ is OWD over $p_1$, either $S_{\va}\subseteq S_1$ or $S_{\va}\subseteq S_2$. Hence, either $S=S_1$ or $S=S_2$ since $S$ is a connected set, and it can not be partitioned into two nonempty subsets which are open. Therefore, $f$ is OWD on $S$, and $f$ is OWD over~$p'$.
\end{proof}

Now the following two theorems follow immediately from the above Propositions. The first one states that the set $$\OWD^j(f)=\{h_j\in\ZZ[\xx_j]|f\text{ is open weak delineable over }h_j\text{ in } \RR^j\}.$$
 is nonempty, so the problem proposed in Section 2 makes sense. 
\begin{thm}\label{thm:owdnepty}
Let $f(\xx_n)\in\ZZ[\xx_n]$, $f$ is OWD over $\Bproj(f,[x_n,\ldots,x_{j}])$ in $\RR^{j-1}$. As a result, the set $\OWD^j(f)$ is nonempty.
\end{thm}
\begin{proof}
By Theorem \ref{thm:McCallum} and Theorem \ref{thm:Brown}, $f$ is OWD over $\Bproj(f,[x_n])$, and $\Bproj(f,[x_n,\ldots,x_{i}])$ is OWD over $\Bproj(f,[x_n,\ldots,x_{i+1}])$ ($j-1\le i\le n$). By Proposition \ref{prop:weaktransitive}, $f$ is OWD over $\Bproj(f,[x_n,\ldots,x_{j}])$ in $\RR^{j-1}$.
\end{proof}
The next theorem says that there is a minimal element in $\OWD^j(f)$ in some sense.
 In the following, we call the number of the open components in $\RR^j$ defined by $p\neq0$
 the {\em scale} of the open weak CAD of $f$ defined by $p$ in $\RR^j$.
\begin{thm}
Let $f(\xx_n)\in\ZZ[\xx_n]$, there exists $p\in\OWD^j(f)$, such that any $p'\in\OWD^j(f)$, $\zero(p)\subseteq\zero(p')$. In particular, the scale of the open weak CAD of $f$ defined by $p$ in $\RR^j$ is minimal.
\end{thm}
\begin{proof}
  Otherwise, for any $p\in\OWD^j(f)$, there exists $p'\in\OWD^j(f)$ such that $\zero(p^2+p'^2)=\zero(p)\cap\zero(p')\subsetneq\zero(p)$. By Proposition \ref{prop:weakgcd}, $p^2+p'^2\in \OWD^j(f)$. Thus, we can find a sequence of polynomials $p_i\in \OWD^j(f)$ ($i=1,2,\ldots$), such that the descending chain of closed sets $\zero(p_1)\supsetneq\zero(p_2)\supsetneq\cdots$ is not stationary, which contradicts with the well-known fact that $\RR^j$ is noetherian under the Zariski Topology.
\end{proof}

We want to obtain an element in $\OWD^j(f)$ as small as possible. A natural way is to apply Theorem \ref{thm:owdnepty} and Proposition \ref{prop:weakgcd}. Let us take $f\in\ZZ[x,y,z,w]$ as an example. According to Theorem \ref{thm:owdnepty}, $f$ is OWD over $p_1=\Bproj(f,[w,z])$ and $p_2=\Bproj(f,[z,w])$ in $\RR^2$, respectively. According to Proposition \ref{prop:weakgcd}, $f$ is OWD over $p'=p_1^2+p_2^2$ in $\RR^2$.
If we want to obtain an element in $\OWD^1(f)$ from $p'$, the simplest way is to apply Brown's projection operator directly, $\Bproj(p',y)\in\OWD^1(f)$. But it is quite possible that $\Bproj(p',y)$ is more complicated than $\Bproj(f,[w,z,y])$, since the degree of $p'$ is twice as much as that of $\Bproj(f,[w,z])$.
For a polynomial $p\in\ZZ[x,y]$, whether $p\in\OWD^2(f)$ or not is only dependent on the real zeros of $p$. It indicates us that, instead of computing $\Bproj(p',y)$ directly, we may find a polynomial $p\in\ZZ[x,y]$, such that $\zero(p)$ and $\zero(p')$ are almost the same, and $\Bproj(p,y)\in\OWD^1(f)$. Roughly speaking, let $p=\gcd(p_1,p_2)$, $p_1'=\frac{p_1}{p}$, $p_2'=\frac{p_2}{p}$. For simplicity, we suppose that $p$ is an irreducible polynomial. It is clear that $\zero(p_1^2+p_2^2)=\zero(p(p_1'^2+p_2'^2))=\zero(p)\cup\zero(p_1'^2+p_2'^2)$. Since $\gcd(p_1',p_2')=1$, intuitively, $\zero(p_1'^2+p_2'^2)=\zero(p_1',p_2')$ is a closed set of codimension $2$. If $p$ is not semi-definite, one can show that $\zero(p)$ is a closed set of codimension $1$. Thus, the two sets $\zero(p)$ and $\zero(p')$ are almost the same. We will show that $f$ is ``almost'' OWD over $\Bproj(p,y)$ in $\RR^1$.

In order to state our results precisely, we introduce the following definitions.

\begin{defn}
Let $Q=\{q_1,\ldots,q_s\}$ be a polynomial set, where $q_i\neq0$. We say that $Q$ is a polynomial set of level $j$ if $q_i\in \RR[x_1,\ldots,x_j]$. Define
$$Q^2=\sum_{i=1}^s q_i^2.$$
Let $\bm{\alpha}\in\RR^j$, define
$$Q(\bm{\alpha})=\{q_i(\bm{\alpha})|i=1,\ldots,s\}.$$
It is clear that $\bm{\alpha}\in\zero(Q)$ if and only if $Q(\bm{\alpha})=\{0\}$.

Let $\bm{x}_{I}=(x_{i_1},\ldots,x_{i_l})$, where $1\le i_1<i_2<\ldots<i_l\le j$. Define
$\coeffs(Q,[\bm{x}_{I}])$ to be the set of all the coefficients of all the polynomials $q_i$ in $Q$ with respect to the indeterminates $\bm{x}_{I}$.

Let $p\in\RR[x_1,\ldots,x_j],$ define
$$pQ=\{pq_1,\ldots,pq_s\}.$$
We say that a polynomial $g\in \RR[x_1,\ldots,x_n]$ is a common factor of $Q$, if $g$ is a factor of $q_i$ for every $i$.

We say that a polynomial set $Q$ of level $j$ has codimension at least two and denote it by $\codim(Q)\ge2$, if for any open connected set $U\subseteq\RR^j$, $U\backslash\zero(Q)$ is still open connected.
\end{defn}
Lemma \ref{lem:codim2} below gives a description of a polynomial set of codimension at least two. Before proving the lemma, we introduce the following result.

\begin{lem}\label{thm:1} \citep{han2016proving}
We have
\begin{enumerate}
\item Let $f$ and $g$ be coprime in $\RR[\xx_n]$. For any connected open set $U$ of $\RR^n$, the open set $V = U\backslash \zero(f,g)$ is also connected.

\item Suppose $f\in \RR[\xx_n]$ is a non-zero squarefree polynomial and $U$ is a connected open set of $\mathbb{R}^{n}$. If $f(\xx_n)$ is semi-definite on $U$, then $U\backslash \zero(f)$ is also a connected open set.
\end{enumerate}
\end{lem}

\begin{lem}\label{lem:connected} Let $f=\gcd(f_1,\ldots,f_m)$ where $f_i\in\RR[\xx_n]$, $i=1,2,\ldots,m.$ Suppose $f$ has no real zeros in a connected open set $U\subseteq\RR^n$, then the open set $V = U\backslash \zero(f_1,\ldots,f_m)$ is also connected.
\end{lem}
\begin{proof}
Without loss of generality, we may assume that $f=1$. If $m=1,$ the result is obvious. The result of case $m=2$ is just the claim of Lemma \ref{thm:1}. For $m\ge3$, let $g=\gcd(f_1,\ldots,f_{m-1})$ and $g_i=f_i/g$ $(i=1,\ldots,m-1)$, then $\gcd(f_m,g)=1$ and $\gcd(g_1,\ldots,g_{m-1})=1$. Let $A=\zero(f_1,\ldots,f_m)$, $B=\zero(g_1,\ldots,g_{m-1})\bigcup \zero(g,f_m)$.
Since $A\subseteq B$, we have $U\backslash B\subseteq U\backslash A$. Notice that the closure of ${U\backslash B}$ equals the closure of ${U\backslash A}$, it suffices to prove that $U\backslash B$ is connected, which follows directly from Lemma \ref{thm:1} and the induction.
\end{proof}

\begin{lem}\label{lem:codim2}
Let $Q=\{q_1,\ldots,q_s\}$ be a polynomial set of level $j$. $\codim(Q)\ge2$ if and only if for any common factor $g$ of $Q$, $g$ is semi-definite on $\RR^n$.
\end{lem}
\begin{proof}
  If $\codim(Q)\ge2$, $U\backslash\zero(Q)$ is connected for any open connected set $U\subseteq \RR^j$. It is obvious that $\zero(g)\subseteq\zero(Q)$, and
  $$U\backslash\zero(Q)\subseteq U\backslash\zero(g).$$
  Notice that the closure of $U\backslash\zero(g)$ equals the closure of $U\backslash\zero(Q)$, $U\backslash\zero(g)$ is open connected. In particular, $\RR^j\backslash\zero(g)$ is open connected, and $g$ is semi-definite on $\RR^j$ since $g$ is sign invariant on $\RR^j\backslash\zero(g)$.

If for any common factor $g$ of $Q$, $g$ is semi-definite on $\RR^n$. Let $q=\gcd(q_1,\ldots,q_s)$, $V=U\backslash\zero(q)$, $q_i=qq_i'$, $Q'=\{q_1',\ldots,q_s'\}$. By assumption, $q$ is semi-definite on $\RR^n$, and $V$ is open connected by Lemma \ref{thm:1}. According to Lemma \ref{lem:connected},
$$U\backslash\zero(Q)=V\backslash\zero(Q')$$
is open connected since $\gcd(q_1',\ldots,q_s')=1$.
\end{proof}

By Lemma \ref{lem:codim2}, any common factor $g$ of a polynomial set $Q$ of codimension at least 2 is semi-definite, by Theorem 4.5.1 in \citep{bochnak2013real}, $\dim(\zero(g))\le j-2$. In fact, we can show that $\dim(\zero(Q))\le j-2$. That's why we call $Q$ has codimension at least 2.

\begin{defn}\label{def:generalweakopendeli}
Let $f\in\RR[x_1,\ldots,x_n]$, $p\in \RR[x_1,\ldots,x_j]$ for $1\le j\le n,$ and $Q=\{q_1,\ldots,q_s\}$ is a polynomial set of level $j$, $q_i\neq0$. We say that $f$ is \emph{OWD over} $p$ \emph{w.r.t.} $Q$ (in $\RR^j$), if for any open connected component $S$ of $p\neq0$ in $\RR^j$, $f$ is OWD on $S\backslash\zero(Q)$. We also say that $f$ is \emph{OWD} over $p$ \emph{in general}.
\end{defn}
\begin{rem}
If $f$ is OWD over $p$, it is clear that $f$ is OWD over $p$ w.r.t. $\{1\}$. If $f$ is OWD over $p$ w.r.t. $Q_1$, and $\zero(Q_1)\subseteq\zero(Q_2)$, by definition, $f$ is OWD over $p$ w.r.t. $Q_2$ since $f$ is OWD on $S\backslash \zero(Q_2)\subseteq S\backslash \zero(Q_1)$ for any open connected component $S$ of $p\neq0$ in $\RR^j$. In particular, if $f$ is OWD over $p$ (w.r.t. $\{1\}$), $f$ is OWD over $p$ w.r.t. $Q$ for any polynomial set $Q$ of level $j$.
\end{rem}

The following lemma shows that the above definition is just a variant of OWD when $\codim(Q)\ge2$. In the rest of this paper, we will switch the two notations freely. 
\begin{lem}\label{lem:owdeq}
  Let $f\in\RR[x_1,\ldots,x_n]$, $p\in \RR[x_1,\ldots,x_j]$, $Q$ is a polynomial set of level $j$. If $f$ is OWD over $p$ w.r.t. $Q$, $f$ is OWD over $pQ^2$. If $\codim(Q)\ge2$, $f$ is OWD over $pQ^2$ if and only if $f$ is OWD over $p$ w.r.t. $Q$. 
\end{lem}
\begin{proof}
If $f$ is OWD over $p$ w.r.t. $Q$. Let $S'$ be any open connected component of $pQ^2\neq0$, there there exists a unique open connected component $S$ of $p\neq0$, such that $S'\subseteq S$. Since $Q^2(S')\neq0$, $S'\subseteq S\backslash\zero(Q)$.  $f$ is OWD on $S'$ since $f$ is OWD on $S\backslash\zero(Q)$.

If $\codim(Q)\ge2$, and $f$ is OWD over $pQ^2$. Let $S$ be any open connected component of $p\neq0$, and $S'$ be any open connected component of $pQ^2\neq0$, such that $S'\subseteq S$. We have $S'\subseteq S\backslash\zero(Q)$. Now $pQ^2(S\backslash\zero(Q))\neq0$, and $S\backslash\zero(Q)$ is open connected since $\codim(Q)\ge2$. Thus, $S\backslash\zero(Q)\subseteq S'$. Hence, $S\backslash\zero(Q)=S'$, and $f$ is OWD on $S\backslash\zero(Q)$.
\end{proof}

The following two theorems are analogous to Proposition \ref{prop:weaktransitive} and Proposition \ref{prop:weakgcd}. 
\begin{thm}\label{thm:weaktran}
  Let $f\in\RR[\xx_n]$, $1\le j<k<n$, $p_1\in\RR[\xx_k]$, $p_2\in\RR[\xx_j]$, $Q_1$ is a polynomial set of level $k$. Suppose $f$ is OWD over $p_1$ w.r.t. $Q_1$, and $p_1$ is OWD over $p_2$. $f$ is OWD over $p_2$ w.r.t. $Q_2=\coeffs(Q_1,[x_{k},\ldots,x_{j+1}])$. Furthermore, if $\codim(Q_1)\ge2$, $\codim(Q_2)\ge2.$  
\end{thm}
\begin{proof}
  Let $S''\subseteq\RR^j$ be any open component of $p_2\neq0$, $S=S''\backslash \zero(Q_2)$. We prove that $f$ is OWD on $S$. Let $U\subseteq\RR^n$ be any open component of $f\neq0$, suppose $S\cap \pi_{j}^n(U)\neq\emptyset$. It is clear that $\pi_j^{k^{-1}}(S)\cap\pi_{k}^n(U)\neq\emptyset$. Let $S'\subseteq\RR^k$ be any open component of $p_1\neq0$, such that $(S'\cap\pi_j^{k^{-1}}(S))\cap\pi_{k}^n(U)\neq\emptyset$. $(S'\backslash\zero(Q_1))\cap\pi_{k}^n(U)\neq\emptyset$ is nonempty since $S'\cap\pi_{k}^n(U)$ is a nonempty open set. Thus, $S'\backslash\zero(Q_1)\subseteq\pi_{k}^n(U)$ since $f$ is OWD over $p_1$ w.r.t. $Q_1$. We only need to show that $S\subseteq \pi_{j}^{k}(S'\backslash\zero(Q_1))$. Since $S\cap\pi_{j}^{k}(S')\neq\emptyset$ and $p_1$ is OWD over $p_2$, $S\subseteq S''\subseteq\pi_{j}^{k}(S')$. For any $\bm{\alpha}\in S$, let $\bm{\beta}\in \RR^{k-j}$ such that $(\bm{\alpha},\bm{\beta})\in S'$. $Q_2^2(\bm{\alpha})\neq0$ implies that $Q_2(\bm{\alpha})\neq\{0\}$. Thus, there exists a polynomial $q\in Q_2$, such that $q(\bm{\alpha},x_{j+1},\ldots,x_k)\neq0$. Let $U_{\bm{\beta}}\subseteq\RR^{k-j}$ be a neighborhood of $\bm{\beta}$, such that $(\bm{\alpha},U_{\bm{\beta}})\in S'$, $Q_1(\bm{\alpha},U_{\bm{\beta}})\neq\{0\}$, and
   $(\bm{\alpha},U_{\bm{\beta}})\cap (S'\backslash\zero(Q_1))\neq\emptyset$. This indicates that $\bm{\alpha}\in\pi_{j}^{k}(S'\backslash\zero(Q_1)).$

 If $\codim(Q_1)\ge2$, let $g$ be any common factor of $Q_2$. It is clear that $g$ is a common factor of $Q_1$. Since $\codim(Q_1)\ge2$, $g$ is semi-definite on $\RR^k$ by Lemma \ref{lem:codim2}. By Lemma \ref{lem:codim2} again, $\codim(Q_2)\ge2$.
The theorem is proved.
\end{proof}
As a special case of Theorem \ref{thm:weaktran}, when $k=j+1$, and the set $Q_=\{q_1\}$ has only one polynomial. $f$ is OWD over $p_2$ w.r.t. $Q_2=\coeffs(\{q_1\},[x_{k}])$. In particular, $f$ is OWD over $p_2$ w.r.t. $\lc(q_1,x_{k})$.

One benefit of Theorem \ref{thm:weaktran} is that we can reduce the computational complexity when we apply Brown's projection operator. Namely, suppose $k=j+1$, $f$ is OWD over $p_1q_1$, where $q_1$ is an irreducible polynomial and is semi-definite on $\RR^k$. If we apply Brown's projection operator directly, $f$ is OWD over $\Bproj(p_1q_1,x_k)$. Now we use Theorem \ref{thm:weaktran} to get a simpler but stronger result. By Lemma \ref{lem:owdeq}, $f$ is OWD over $p_1$ w.r.t. $q_1$. According to Theorem \ref{thm:weaktran}, $f$ is OWD over $\Bproj(p_1,x_k)$ w.r.t. $\lc(q_1,x_{k})$. By Lemma \ref{lem:owdeq} again, $f$ is OWD over $\lc(q_1,x_{k})\Bproj(p_1,x_k),$ which is a factor of $\Bproj(p_1q_1,x_k)$.

\begin{thm}\label{thm:weakgcd}
   Let $f\in\RR[\xx_n]$, $p_i\in\RR[\xx_j]$, $Q_i$ is a polynomial set of level $j$ ($i=1,2,\ldots,t$), $p=\gcd(p_1,\ldots,p_t)$, $p_i'=\frac{p_i}{p}$. Suppose $f$ is OWD over $p_i$ w.r.t. $Q_i$. $f$ is OWD over $p$ w.r.t. $Q=\bigcup_{i=1}^tp_i'Q_i$, and $\zero(pQ^2)\subseteq\zero(p_iQ_i^2)$ for any $i$. Furthermore, if $\codim(Q_i)\ge2$, $\codim(Q)\ge2$.
\end{thm}
\begin{proof}
Let $S$ be any open connected component defined by $p\neq0$, $\mathcal{A}_i$ be the set of open components of $p_i\neq0$ in $S$. By definition, $S\backslash{\zero(p_i)}=\bigcup_{S_{i}\in\mathcal{A}_i} S_{i}$, and

Let $U\subseteq\RR^n$ be any open connected set defined by $f\neq0$, $$\mathcal{B}_i=\{S_{i}|S_{i}\in\mathcal{A}_i, S_{i}\backslash{\zero(Q_i)}\subseteq \pi_{j}^n(U)\},$$
$$\mathcal{C}_i=\{S_{i}|S_{i}\in\mathcal{A}_i, S_{i}\backslash{\zero(Q_i)}\cap \pi_{j}^n(U)=\emptyset\}.$$
Since $f$ is OWD over $p_i$ w.r.t. $Q_i$, $\mathcal{A}_i=\mathcal{B}_i\cup \mathcal{C}_i$.
For any $S'\in \mathcal{B}_i$, and $S''\in \mathcal{C}_j$, $S'\cap S''=\emptyset$ since $$(S'\backslash{\zero(Q_i)})\cap (S''\backslash{\zero(Q_j)})\subseteq \pi_{j}^n(U)\cap (S''\backslash{\zero(Q_j)})=\emptyset.$$
Let $$B=\bigcup_{1\le i\le t,S_{i}\in\mathcal{B}_i} S_{i}, B'=\bigcup_{1\le i\le t,S_{i}\in\mathcal{B}_i} (S_{i}\backslash{\zero(Q_i)}),$$
$$C=\bigcup_{1\le i\le t,S_{i}\in\mathcal{C}_i} S_{i}, C'=\bigcup_{1\le i\le t,S_{i}\in\mathcal{C}_i} (S_{i}\backslash{\zero(Q_i)}).$$
We have $B\cap C=\emptyset$, and
$$S\backslash{\zero(p_1,\ldots,p_t)}=\bigcup_{1\le i\le t} S\backslash{\zero(p_i)}=B\cup C.$$
By Lemma $\ref{lem:connected}$, $V=S\backslash{\zero(p_1,\ldots,p_t)}$ is open connected, and can not be
partitioned into two nonempty subsets which are open. Hence, either $B=\emptyset$ or $C=\emptyset$.
Since $$S\backslash{\zero(Q)}=\bigcup_{1\le i\le t} S\backslash{\zero(p_i'Q_i)}=\bigcup_{1\le i\le t} S\backslash{\zero(p_iQ_i)}=B'\cup C',$$
and $B'\subseteq B$, $C'\subseteq C$, either $B'=\emptyset$ or $C'=\emptyset$. Thus, either $S\backslash{\zero(Q)}\subseteq \pi_{j}^n(U)$ or $(S\backslash{\zero(Q)})\cap \pi_{j}^n(U)=\emptyset$. Therefore, $f$ is OWD over $p$ w.r.t. $Q$.

We have
$$\zero(pQ^2)=\zero(p)\cup\zero(Q^2)\subseteq\zero(p)\cup\zero(p_i'Q_i^2)=\zero(p_iQ_i^2),$$
for any $i$.

Since $\gcd(p_1',\ldots,p_t')=1$, any common factor $g$ of $Q$ must be a common factor of $Q_i$ for some $i$. If $\codim(Q_i)\ge2$, by Lemma \ref{lem:codim2}, $g$ is semi-definite on $\RR^j$. By Lemma \ref{lem:codim2} again, $\codim(Q)\ge2$.
\end{proof}
We can apply Theorem \ref{thm:weaktran}, Theorem \ref{thm:weakgcd} and Brown's projection operator $\Bproj$ to get a weak open CAD with ``smaller'' scale now. let us take $f\in\ZZ[x,y,z,w]$ again as an example.
According to Theorem \ref{thm:owdnepty}, $f$ is OWD over $\Hproj(f,[w,z],z)=\Bproj(f,[w,z])$ and $\Hproj(f,[w,z],w)=\Bproj(f,[z,w])$ in $\RR^2$, respectively. Let $$\Hproj(f,[w,z])=\gcd(\Hproj(f,[w,z],z),\Hproj(f,[w,z],w)),$$ $${\Hproj^{\vartriangle}(f,[w,z],w)}=\frac{\Hproj(f,[w,z],w)}{\Hproj(f,[w,z])},$$ $${\Hproj^{\vartriangle}(f,[w,z],z)}=\frac{\Hproj(f,[w,z],z)}{\Hproj(f,[w,z])}.$$
According to Theorem \ref{prop:weakgcd}, $f$ is OWD over $\Hproj(f,[w,z])$ w.r.t. $${\Hproj^{\ast}(f,[w,z])}=\{{\Hproj^{\vartriangle}(f,[w,z],w)},{\Hproj^{\vartriangle}(f,[w,z],z)}\}$$
in $\RR^2$, and $f$ is OWD over $h_2(f)=\Hproj(f,[w,z]){\Hproj^{\ast}(f,[w,z])}^2$.

According to Theorem \ref{thm:weaktran}, $f$ is OWD over $$\Hproj(f,[w,z,y],y)=\Bproj(\Hproj(f,[w,z]),y)$$
w.r.t. $${\Hproj^{\ast}(f,[w,z,y],y)}=\coeffs({\Hproj^{\ast}(f,[w,z])},[y])$$
in $\RR$.

Similarly, we can define
$$\Hproj(f,[w,z,y],z),\, \Hproj(f,[w,z,y],w),\, \Hproj(f,[w,z,y]),$$
$${\Hproj^{\ast}(f,[w,z,y],z)},\,\Hproj^{\ast}(f,[w,z,y],w),$$
and
$${\Hproj^{\vartriangle}(f,[w,z,y],y)},{\Hproj^{\vartriangle}(f,[w,z,y],z)},\,\Hproj^{\vartriangle}(f,[w,z,y],w).$$
Define
\begin{align*}
\Hproj^{\ast}(f,[w,z,y])=&\Hproj^{\vartriangle}(f,[w,z,y],y)\Hproj^{\ast}(f,[w,z,y],y)\bigcup\\
& \Hproj^{\vartriangle}(f,[w,z,y],z)\Hproj^{\ast}(f,[w,z,y],z)\bigcup\\
&\Hproj^{\vartriangle}(f,[w,z,y],w)\Hproj^{\ast}(f,[w,z,y],w)
\end{align*}
By applying Theorem \ref{thm:weakgcd} again, one can show that $f$ is OWD over $\Hproj(f,[w,z,y])$ w.r.t. $\Hproj^{\ast}(f,[w,z,y])$, or equivalently, $f$ is OWD over
$$h_1(f)=\Hproj(f,[w,z,y])\Hproj^{\ast}(f,[w,z,y])^2.$$

This procedure could apply to any polynomial $f\in\ZZ[\xx_n]$. In order to present our results in general, we need to introduce \emph{open weak CAD projection operator} $\Hproj$.
\begin{defn}[Open weak CAD projection operator]\label{def:hp}
Let $f\in \ZZ[x_1,\ldots,x_n]$. 
For given $m (1\le m\le n),$ denote $[\yy]=[y_1,\dots,y_{m}]$ where $ y_i \in \{x_1,\dots,x_n\}$ for $1\le i\le m$ and $y_i\neq y_j$ for $i\neq j$. For $1\le i\le m$,
$\Hproj(f,[\yy],y_i)$ and $\Hproj(f,[\yy])$ are defined recursively as follows.
        \begin{align*}
                 \Bproj(f,[y_i])    &=\Res(\sqrfree(f),\frac{\partial \sqrfree(f)}{\partial y_{i}}, y_{i}),\\
                \Hproj(f,[\yy],y_i) &=\Bproj(\Hproj(f,[\hat{\yy}_i],[y_i]),\\
                \Hproj(f,[\yy])     &=\gcd(\Hproj(f,[\yy],y_1),\ldots,\Hproj(f,[\yy],y_m)),\\
                \Hproj(f,[~])       &=f,
    \end{align*}
where $\hat{[\yy]}_i=[y_1,\ldots,y_{i-1},y_{i+1},\ldots,y_{m}]$.

We define ${\Hproj^{\vartriangle}(f,[\yy],y_i)}$, $\Hproj^{\ast}(f,[\yy],y_i)$, $\Hproj^{\ast}(f,[\yy])$ recursively as follows.
\begin{align*}
\Hproj^{\ast}(f,[~])&=\{1\},\\
{\Hproj^{\ast}(f,[\yy],y_i)}&=\coeffs({\Hproj^{\ast}(f,[\hat{\yy}_i])},y_i),\\
{\Hproj^{\vartriangle}(f,[\yy],y_i)}&=\frac{\Hproj(f,[\yy],y_i)}{\Hproj(f,[\yy])},\\
{\Hproj^{\ast}(f,[\yy])}&=\bigcup_{i=1}^m {\Hproj^{\vartriangle}(f,[\yy],y_i)}{\Hproj^{\ast}(f,[\yy],y_i)}.
\end{align*}
\end{defn}

\begin{ex}\label{ex:step}
We have
\begin{align*}
\Hproj(f,\left[  x_{1},x_{2}\right]  )  & =\gcd\left(  \Hproj\left(  f,\left[
x_{1},x_{2}\right]  ,x_{1}\right)  ,\Hproj\left(  f,\left[  x_{1}%
,x_{2}\right]  ,x_{2}\right)  \right),  \\
\Hproj\left(  f,\left[  x_{1},x_{2}\right]  ,x_{1}\right)    &
=\Bproj(\Hproj(f,\left[  x_{2}\right]  ),[x_{1}]),\\
\Hproj\left(  f,\left[  x_{1},x_{2}\right]  ,x_{2}\right)    &
=\Bproj(\Hproj(f,\left[  x_{1}\right]  ),[x_{2}]),\\
\Hproj(f,\left[  x_{2}\right]  )  & =\Hproj(f,[x_{2}],x_{2}),\\
\Hproj(f,\left[  x_{1}\right]  )  & =\Hproj(f,[x_{1}],x_{1}),\\
\Hproj(f,[x_{2}],x_{2})  & =\gcd(\Bproj(\Hproj(f,[~]),[x_{2}])),\\
\Hproj(f,[x_{1}],x_{1})  & =\gcd(\Bproj(\Hproj(f,[~]),[x_{1}])),\\
\Hproj(f,[~])  & =f.
\end{align*}
Condensing the above expressions, we have%
\[
\Hproj(f,\left[  x_{1},x_{2}\right]  )=\gcd\left(  \Bproj(\Bproj(f,[x_{2}%
]),[x_{1}]),\Bproj(\Bproj(f,[x_{1}]),[x_{2}])\right).
\]
We have
\begin{align*}
{\Hproj^{\ast}(f,[~])}&=\{1\},\\
{\Hproj^{\ast}(f,[x_1],x_1)}&=\coeffs({\Hproj^{\ast}(f,[~])},x_1)=\{1\},\\
{\Hproj^{\vartriangle}(f,[x_1],x_1)}&=\frac{\Hproj(f,[x_1],x_1)}{\Hproj(f,[x_1])}=1,\\
{\Hproj^{\ast}(f,[x_1])}&={\Hproj^{\vartriangle}(f,[x_1],x_1)}{\Hproj^{\ast}(f,[x_1],x_1)}=\{1\}.
\end{align*}
Similarly,
$${\Hproj^{\ast}(f,[x_2])}={\Hproj^{\vartriangle}(f,[x_2],x_2)}{\Hproj^{\ast}(f,[x_2],x_2)}=\{1\}.$$
\begin{align*}
{\Hproj^{\ast}(f,[x_2,x_1],x_2)}&=\coeffs({\Hproj^{\ast}(f,[x_1])},x_2)=\{1\},\\
{\Hproj^{\ast}(f,[x_2,x_1],x_1)}&=\coeffs({\Hproj^{\ast}(f,[x_2])},x_1)=\{1\},\\
{\Hproj^{\vartriangle}(f,[x_2,x_1],x_1)}&=\frac{\Hproj(f,[x_2,x_1],x_1)}{\Hproj(f,[x_2,x_1])},\\
{\Hproj^{\vartriangle}(f,[x_2,x_1],x_2)}&=\frac{\Hproj(f,[x_2,x_1],x_2)}{\Hproj(f,[x_2,x_1])},\\
{\Hproj^{\ast}(f,[x_2,x_1])}&=\bigcup_{i\in\{1,2\}}{\Hproj^{\vartriangle}(f,[x_2,x_1],x_i)}{\Hproj^{\ast}(f,[x_2,x_1],x_i)}\\
&=\{{\Hproj^{\vartriangle}(f,[x_2,x_1],x_1)},{\Hproj^{\vartriangle}(f,[x_2,x_1],x_2)}\}.
\end{align*}
\end{ex}
The Algorithm \ref{alg:openweakcad} ({\tt Projection phase of open weak CAD}) based on the open weak CAD projection operator $\Hproj$ solves the problem proposed in Section~\ref{sec:problem}.

\begin{algorithm}[!ht]
\caption{{Projection polynomials of open weak CAD}}
\label{alg:openweakcad}
\medskip
\begin{description}[leftmargin=3em,style=nextline,itemsep=0.5em]
\item[\sf In:]   $f\in \ZZ[x_1,\ldots,x_n]$
\item[\sf Out:]  $h_1, \ldots, h_{n-1}$ where $h_j\in \ZZ[x_1,\ldots,x_j]$
such that $f$ is open weak delineable over $h_j$ in $\RR^j$.
\item[\sf 1:]  For all $1\le j\le n-1$, compute
\[
\Hproj(f,[x_n,\ldots,x_{j+1}]) \;\;\;\;\;\; \text{and}\;\;\;\;\;\;
{\Hproj^{\ast}(f,[x_n,\ldots,x_{j+1}])}
\]
using Definition \ref{def:hp}.
\item[\sf 2:]  For all $1\le j \le n-1$, compute       \[h_j(f)=\Hproj(f,[x_n,\ldots,x_{j+1}])\cdot{\Hproj^{\ast}(f,[x_n,\ldots,x_{j+1}])}^2.\]
\end{description}
\end{algorithm}
\begin{ex}\label{ex:1-1}
We illustrate the Algorithm \ref{alg:openweakcad} using the polynomial $f$ from Example \ref{ex:1}.
\medskip
\begin{description}[leftmargin=3em,style=nextline,itemsep=0.5em]
\item[\sf In:]   $f=(x_3^2+x_2^2+x_1^2-1)(4x_3+3x_2+2x_1-1)\in \ZZ[x_1, x_2,x_3]$

\item[\sf 1:] Note
\[\begin{array}{ll}
\Hproj(f,[x_3])&=\Hproj(f,[x_3],x_3)=\Bproj(\Hproj(f,[~]),[x_3])=\Bproj(f,[x_3])\\
& =(x_2^2+x_1^2-1)(25x_2^2+12x_2x_1+20x_1^2-6x_2-4x_1-15),\\
\Hproj(f,[x_3, x_2],x_2) & =\Bproj(\Hproj(f,[x_3]),[x_2])=\Bproj(\Hproj(f,[x_3],x_3), [x_2])\\
  & = (x_1-1)(x_1+1)(29x_1^2-4x_1-24)(13x_1^2-4x_1-8), \\
\Hproj(f,[x_3, x_2],x_3) & = \Bproj(\Hproj(f,[x_2]),[x_3])=\Bproj(\Hproj(f,[x_2],x_2), [x_3])\\
  & = (x_1-1)(x_1+1)(29x_1^2-4x_1-24)(20x_1^2-4x_1-15),\\
\Hproj(f,[x_3,x_2])&=\gcd(\Hproj(f,[x_3, x_2],x_2),\Hproj(f,[x_3, x_2],x_3))\\
&=(x_1-1)(x_1+1)(29x_1^2-4x_1-24). \\
\\
{\Hproj^{\ast}(f,[x_3])}&={\Hproj^{\vartriangle}(f,[x_3],x_3)}{\Hproj^{\ast}(f,[x_3],x_3)}=\{1\},\\
{\Hproj^{\vartriangle}(f,[x_3,x_2],x_2)}&=\frac{\Hproj(f,[x_3,x_2],x_2)}{\Hproj(f,[x_3,x_2])}=13x_1^2-4x_1-8,\\
{\Hproj^{\vartriangle}(f,[x_3,x_2],x_3)}&=\frac{\Hproj(f,[x_3,x_2],x_3)}{\Hproj(f,[x_3,x_2])}=20x_1^2-4x_1-15,\\
{\Hproj^{\ast}(f,[x_3,x_2])}&=\{{\Hproj^{\vartriangle}(f,[x_3,x_2],x_2)},{\Hproj^{\vartriangle}(f,[x_3,x_2],x_3)}\}\\
&=\{13x_1^2-4x_1-8,20x_1^2-4x_1-15\}.
\end{array}\]
\item[\sf 2:] Note
\[\begin{array}{ll}
h_1(f) &=\Hproj(f,[x_3,x_2]){\Hproj^{\ast}(f,[x_3,x_2])}^2\\
    & =(x_1-1)(x_1+1)(29x_1^2-4x_1-24)((20x_1^2-4x_1-15)^2+(13x_1^2-4x_1-8)^2),\\
h_2(f) &=\Hproj(f,[x_3]){\Hproj^{\ast}(f,[x_3])}^2 \\
    &=(x_2^2+x_1^2-1)(25x_2^2+12x_2x_1+20x_1^2-6x_2-4x_1-15).
\end{array}
\]
\item[\sf Out:]  $h_1(f) =(x_1-1)(x_1+1)(29x_1^2-4x_1-24)((20x_1^2-4x_1-15)^2+(13x_1^2-4x_1-8)^2),$
\item[]          $h_2(f) =(x_2^2+x_1^2-1)(25x_2^2+12x_2x_1+20x_1^2-6x_2-4x_1-15).$
\end{description}
\end{ex}
\begin{rem}
  If $\zero({\Hproj^{\ast}(f,[x_n,\ldots,x_{j+1}])})=\emptyset$, we can get a simpler expression of $h_j(f)$ by computing $h_j(f)=\Hproj(f,[x_n,\ldots,x_{j+1}])$ instead of computing $$h_j(f)=\Hproj(f,[x_n,\ldots,x_{j+1}])\cdot{\Hproj^{\ast}(f,[x_n,\ldots,x_{j+1}])}^2,$$ since they have the same real zeros. When $n=3$ and $j=1$, it is always the case, since ${\Hproj^{\vartriangle}(f,[x_3,x_2],x_2)},{\Hproj^{\vartriangle}(f,[x_3,x_2],x_3)}$ are two coprime polynomials of one variable and $\zero({\Hproj^{\ast}(f,[x_3,x_{2}])})=\zero({\Hproj^{\vartriangle}(f,[x_3,x_2],x_2)},{\Hproj^{\vartriangle}(f,[x_3,x_2],x_3)})=\emptyset.$ Hence, in the above example, we can get a simper expression of $h_1(f)$,
  $$h_1(f) =(x_1-1)(x_1+1)(29x_1^2-4x_1-24).$$
\end{rem}
\begin{thm}[Correctness]\label{thm:openweakcad} Algorithm \ref{alg:openweakcad} is correct.
\end{thm}
\begin{proof}
Let $f\in\ZZ[\xx_n]$. We begin by proving that  $f$ is OWD over $\Hproj(f,[x_n,\ldots,x_{k+1}])$ w.r.t. ${\Hproj^{\ast}(f,[x_n,\ldots,x_{k+1}])}$, and $\codim({\Hproj^{\ast}(f,[x_n,\ldots,x_{k+1}])})\ge2$ by induction on $k$.

When $k=n-1$, $f$ is OWD over $\Hproj(f,[x_n])=\Bproj(f,[x_n])$ w.r.t. ${\Hproj^{\ast}(f,[x_n])}=\{1\}$, and $\codim({\Hproj^{\ast}(f,[x_n])})\ge2$.  Suppose the theorem is true for $k=n-1,\ldots,j+1$. Now, we consider the case $k=j$. By Theorem \ref{thm:weaktran} and the induction, $f$ is OWD over $\Hproj(f,[x_n,\ldots,x_{j}],x_t)$ w.r.t. ${\Hproj^{\ast}(f,[x_n,\ldots,x_{j}],x_t)}$, and $\codim({\Hproj^{\ast}(f,[x_n,\ldots,x_{j}],x_t)})\ge2$ for $t=j,\ldots,n$. By Theorem \ref{thm:weakgcd}, $f$ is OWD over $\Hproj(f,[x_n,\ldots,x_{j}])$ w.r.t. ${\Hproj^{\ast}(f,[x_n,\ldots,x_{j}])}$, and $\codim({\Hproj^{\ast}(f,[x_n,\ldots,x_{j}])})\ge2.$ We complete the induction.

By Lemma \ref{lem:owdeq}, $f$ is OWD over $h_{i-1}(f)$.
Hence the algorithm is correct.
\end{proof}

Although the expression of $h_j(f)$ in Algorithm \ref{alg:openweakcad} is complicated, the zero of $h_j(f)$ is simper than that of $\Bproj(f,[x_n,\ldots,x_{j+1}])$.
\begin{thm}\label{thm:smhp}Let $f\in\ZZ[x_1,\ldots,x_n]$, we have
$$\Hproj(f,[x_n,\ldots,x_{k+1}])|\Hproj(f,[x_n,\ldots,x_{k+1}],x_{k+1})|\Bproj(f,[x_n,\ldots,x_{k+1}]),$$
and there exists a polynomial $q_k\in {\Hproj^{\ast}(f,[x_n,\ldots,x_{k+1}])}$, such that $$q_k|\Bproj(f,[x_n,\ldots,x_{k+1}]);$$
Moreover, in Algorithm \ref{alg:openweakcad},
\begin{align*}
\zero(h_k)=&\zero(\Hproj(f,[x_n,\ldots,x_{k+1}]))\cup\zero({\Hproj^{\ast}(f,[x_n,\ldots,x_{k+1}])})\\
\subseteq&\zero(\Bproj(f,[x_n,\ldots,x_{k+1}])).
\end{align*}
\end{thm}
\begin{proof}
We prove the theorem by induction on $k$. When $k=n-1$,
$$\Hproj(f,[x_n])=\Hproj(f,[x_n],x_n)=\Bproj(f,[x_n]),\,{\Hproj^{\ast}(f,[x_n])}=\{1\}.$$
Let $q_{n-1}=1\in {\Hproj^{\ast}(f,[x_n])}$, $q_{n-1}|\Hproj(f,[x_n])$. By definition, $$h_{n-1}=\Hproj(f,[x_n]){\Hproj^{\ast}(f,[x_n])}^2=\Hproj(f,[x_n]),\, \zero(h_{n-1})=\zero(\Bproj(f,[x_n])).$$
Suppose the theorem is true for $k=n-1,\ldots,j+1$. Now, we consider the case $k=j$. By definition and the induction, $\Hproj(f,[x_n,\ldots,x_{j+1}])$ is a factor of $$\Hproj(f,[x_n,\ldots,x_{j+1}],x_{j+1})=\Bproj(\Hproj(f,[x_n,\ldots,x_{j+2}]),[x_{j+1}]),$$
and
$$\Hproj(f,[x_n,\ldots,x_{j+2}])|\Bproj(f,[x_n,\ldots,x_{j+2}]).$$
These imply that
$\Bproj(\Hproj(f,[x_n,\ldots,x_{j+2}]),[x_{j+1}])$ is a factor of $$\Bproj(\Bproj(f,[x_n,\ldots,x_{j+2}]),[x_{j+1}])=\Bproj(f,[x_n,\ldots,x_{j+1}]),$$ and
$$\Hproj(f,[x_n,\ldots,x_{j+1}])|\Bproj(f,[x_n,\ldots,x_{j+1}]).$$
 Suppose $q_{j+1}\in {\Hproj^{\ast}(f,[x_n,\ldots,x_{j+2}],x_{j+2})}$, such that
 $$q_{j+1}|\Bproj(f,[x_n,\ldots,x_{j+2}]).$$
 Let $q_j=\lc(q_{j+1},x_{j+1})\cdot{\Hproj^{\vartriangle}(f,[x_n,\ldots,x_{j+1}],x_{j+1})}\in {\Hproj^{\ast}(f,[x_n,\ldots,x_{j+1}])}$.
 Now,
  $$\lc(q_{j+1},x_{j+1})|\lc(\Bproj(f,[x_n,\ldots,x_{j+2}]),x_{j+1})|\Bproj(f,[x_n,\ldots,x_{j+1}]).$$
 and
$${\Hproj^{\vartriangle}(f,[x_n,\ldots,x_{j+1}],x_{j+1})}|\Hproj(f,[x_n,\ldots,x_{j+1}],x_{j+1})|\Bproj(f,[x_n,\ldots,x_{j+1}]).$$
Thus, $q_j|\Bproj(f,[x_n,\ldots,x_{j+1}])$.

Now, we have
\begin{align*}
&\zero({\Hproj^{\ast}(f,[x_n,\ldots,x_{j+1}])})\\
\subseteq &\zero(q_j)=\zero(\lc(q_{j+1},x_{j+1}))\cup\zero({\Hproj^{\vartriangle}(f,[x_n,\ldots,x_{j+1}],x_{j+1})})\\
\subseteq&\zero(\Bproj(f,[x_n,\ldots,x_{j+1}])).
\end{align*}
We complete the induction, and the theorem is proved.
\end{proof}
 This theorem implies that, for every open cell $C'$ of open CAD produced by Brown's projection $\Bproj(f,[x_n,\ldots,x_{j+1}])$, there exists an open cell $C$ of open weak CAD produced by $h_j$ such that $C'\subseteq C$. Thus, the scale of open weak CAD is not bigger than that of open CAD.

\begin{rem}
In Algorithm \ref{alg:openweakcad}, the scale of the open weak CAD of $f$ defined by $h_j$ in $\RR^j$ is not always the smallest. For example, let $f$ be the polynomial in Example \ref{ex:wod}, then $h_1=(x_1-1)x_1$, and $f$ is open weak delineable over $x_1$, as mentioned earlier.
\end{rem}

\section{Application: Open Sample}
\label{sec:reduced_open_cad}
As a first application of Theorem \ref{thm:openweakcad}, we show how to compute open sample based on Algorithm \ref{alg:openweakcad}.

\begin{defn} $($Open sample$)$
A set of sample points $S_f\subseteq \RR^k\setminus \zero(f)$ is said to be an {\em open sample} defined by $f(\xx_k)\in\ZZ[\xx_k]$ in $\RR^k$ if it has the following property: for every open connected set $U\subseteq \RR^k$ defined by $f\neq0$, $S_f\cap U\ne \emptyset$.

Suppose $g(\xx_k)$ is another polynomial. If $S_f$ is an open sample defined by $f(\xx_k)$ in $\RR^k$ such that $g(\va)\neq0$ for any $\va\in S_f$, then we denote the open sample by $S_{f,g}$. 
\end{defn}

As a corollary of Theorems \ref{thm:McCallum} and \ref{thm:Brown}, a property of open CAD (or GCAD) is that at least one sample point can be taken from every highest dimensional cell via the open CAD (or GCAD) lifting phase.
So, an open CAD is indeed an open sample.



Obviously, there are various efficient ways to compute $S_{f,g}$ for two given {\em univariate} polynomials $f,g\in \ZZ[x]$. For example, we may choose one rational point from every open interval defined by the real zeros of $f$ such that $g$ does not vanish at this point. Therefore, we only describe the specification of such algorithms {\tt  SPOne} here and omit the details of the algorithms.
\begin{defn}
Let $\va_j=(\alpha_1,\ldots,\alpha_j)\in \RR^j$ and $S\subseteq \RR$ be a finite set, define
\[\va_j\boxplus S=\{(\alpha_1,\ldots,\alpha_j,\beta)\mid \beta\in S\}.\]
\end{defn}

\begin{algorithm}[!ht]
        \caption{{\tt SPOne}} \label{one-dimsample}
        \begin{algorithmic}[1]
                \Require{Two univariate polynomials $f,g \in \ZZ[x]$.}
                \Ensure{$S_{f,g}$, an open sample defined by $f(x)$ in $\RR$ such that $g(\va)\neq0$ for any $\va\in S_{f,g}$.}
        \end{algorithmic}
\end{algorithm}

\begin{algorithm}[!ht]
        \caption{{\tt OpenSP}} \label{opensample}
        \begin{algorithmic}[1]
                \Require{Two lists of polynomials $L_1=[f_n(\xx_n), \ldots, f_j(\xx_j)]$, $L_2=[g_n(\xx_n), \ldots, g_j(\xx_j)]$, and a set of points $S$ in $\RR^{j}$.}
                \Ensure{A set of sample points in $\RR^{n}$.}
        \State $P_j:=S$
        \For {$i$ from $j+1$ to $n$}
        \State $P_{i}:=\emptyset$
        \For {$\va$ in $O$}
        \State $P_i:=P_i\bigcup(\va\boxplus {\tt SPOne}(f_i(\va,x_i),g_i(\va,x_i)))$ 
        \EndFor
        \EndFor
        \State \Return $P_n$
        \end{algorithmic}
\end{algorithm}

\begin{rem}
For a polynomial $f(\xx_n)\in \ZZ[\xx_n]$,
let
\begin{align*}
B_1 & =[f, \Bproj(f,[x_n]),\dots, \Bproj(f,[x_n,\dots,x_2])],\\
B_2 &=[1,\ldots,1],\\
S   & ={\tt SPOne}(\Bproj(f,[x_n,\dots,x_2]),1),
\end{align*}
then ${\tt OpenSP}(B_1,B_2,S)$ is an open CAD (an open sample) defined by $f(\xx_n)$.
\end{rem}

\begin{rem}
The output of ${\tt OpenSP}(L_1,L_2,S)$ is dependent on the method of choosing sample points in Algorithm $\tt SPOne$. In the following, when we use the terminology ``any ${\tt OpenSP}(L_1,L_2,S)$'', we mean ``no matter which method is used in Algorithm $\tt SPOne$ for choosing sample points".
\end{rem}

\begin{defn}
Given a polynomial $f\in\ZZ[x_1,\ldots,x_n]$ and an ordering $x_1\prec x_2\prec\ldots\prec x_n$.
Let $\Hproj(f,n+1)=1$, we define polynomials $\Hproj(f,i)$$(2\le i\le n)$ recursively as follows
 $$\Hproj(f,i)=\lc(\Hproj(f,i+1),x_{i}){\Hproj^{\vartriangle}(f,[x_n,\ldots,x_{i}],x_{i})}.$$
 Define two polynomial sets,
\begin{align*}
\overline{\Hproj}(f,i)  &=\{f,\Hproj(f,[x_{n}]), \ldots, \Hproj(f,[x_{n},\ldots,x_i])\},\\
{\Hproj^{\vartriangle}}(f,i) &=\{\Hproj(f,n+1),\Hproj(f,n),\ldots,\Hproj(f,i)\}.
\end{align*}
\end{defn}
\begin{thm}\label{thm:owdsp}
Let $f\in\ZZ[\xx_n]$. $f$ is OWD over $\Hproj(f,[x_n,\ldots,x_i])$ w.r.t. $\Hproj(f,i)$.
\end{thm}
\begin{proof}
  By definition, $\Hproj(f,i)\in{\Hproj^{\ast}(f,[x_n,\ldots,x_i])}$,
$$\zero({\Hproj^{\ast}(f,[x_n,\ldots,x_i])})\subseteq\zero(\Hproj(f,i)).$$
 According to Theorem \ref{thm:openweakcad}, $f$ is OWD over $\Hproj(f,[x_n,\ldots,x_i])$ w.r.t. ${\Hproj^{\ast}(f,[x_n,\ldots,x_{i}])}$. Hence, $f$ is OWD over $\Hproj(f,[x_n,\ldots,x_i])$ w.r.t. $\Hproj(f,i)$.
\end{proof}

\begin{defn}\label{de:reduced}
A {\em reduced open CAD} of $f(\xx_n)$ w.r.t. $[x_n,\ldots,x_{j+1}]$ is a set of sample points in $\RR^n$,
 $${\tt OpenSP}(\overline{\Hproj}(f,i),{\Hproj^{\vartriangle}}(f,i),OS)$$
 where $OS$ is an open sample $OS=S_{\Hproj(f,[x_n,\ldots,x_{j+1}]),\Hproj(f,j+1)}$ in $\RR^{j}$.
\end{defn}

\begin{thm}\label{thm:reducedcad}
Any reduced open CAD of $f(\xx_n)$ w.r.t. $[x_n,\ldots,x_{j+1}]$ is an open sample defined by $f(\xx_n)$.
\end{thm}

\begin{proof}
Let $P_j$ be defined as in Algorithm \ref{opensample}. For any open connected component $U\subseteq\RR^n$ defined by $f\neq0$, we prove by induction on $k$ that $\pi_{k}^n(U)\cap P_k\neq\emptyset$.
When $k=j$, let $S\subseteq\RR^j$ be an open connected component of $\Hproj(f,[x_n,\ldots,x_{j+1}])\neq0$ such that $S\cap \pi_{j}^n(U)\neq\emptyset$. By Theorem \ref{thm:owdsp}, $f$ is OWD over $\Hproj(f,[x_n,\ldots,x_{j+1}])$ w.r.t. $\Hproj(f,j+1)$, $S\backslash\zero(\Hproj(f,j+1))\subseteq \pi_{j}^n(U)$. Let $\va \in P_j$, such that $\va\in S\backslash\zero(\Hproj(f,j+1))$, $\va\in \pi_{j}^n(U)\cap P_j$.

Suppose the induction holds for $k=j,j+1,\ldots,i$. Now, we consider the case $k=i+1$. Let $\va\in P_i$ such that $\va\in \pi_{i}^n(U)$. $\Hproj(f,i+1)(\va)\neq0$ implies that $$\Hproj(f,[x_n,\ldots,x_{i+1}],x_{i+1})(\va)\neq0.$$
Let $S\subseteq\RR^i$ be the open component of $\Hproj(f,[x_n,\ldots,x_{i+1}],x_{i+1})\neq0$ containing $\va$.
 Since $\va\in S\cap\pi_{i+1}^n(U)$, $S\cap\pi_{i}^n(U)\neq \emptyset$ and the open set ${\pi_{i}^{i+1}}^{-1}(S)\cap\pi_{i+1}^n(U)$ is nonempty. Let $S'$ be an open connected component of $\Hproj(f,[x_n,\ldots,x_{i+2}])\neq0$ (for simplicity, we define $\Hproj(f,[x_n,\ldots,x_{i+2}])=f$ if $i=n-1$) such that $S'\cap {\pi_{i}^{i+1}}^{-1}(S)\cap\pi_{i+1}^n(U)\neq\emptyset$. By Theorem \ref{thm:owdsp}, $f$ is OWD over $\Hproj(f,[x_n,\ldots,x_{i+2}])$ w.r.t. $\Hproj(f,i+2)$, and
 $$S'\backslash\zero(\Hproj(f,i+2))\subseteq\pi_{i+1}^n(U).$$
 Since $\Hproj(f,[x_n,\ldots,x_{i+2}])$ is OWD over $\Hproj(f,[x_n,\ldots,x_{i+1}],x_{i+1})$, $S\subseteq\pi_{i}^{i+1}(S')$ and $\va\in\pi_{i}^{i+1}(S')$. Let $U_{\va}\subseteq\RR$ be the maximal open set such that $(\va,U_{\va})\in S'$. $\Hproj(f,i+1)(\va)\neq0$ implies that $\lc(\Hproj(f,i+2),x_{i+1})(\va)\neq0$, and $\Hproj(f,i+2)(\va,x_{i+1})$ is a nonzero polynomial. Thus, there exists
$$\beta\in U_{\va}\cap{\tt SPOne}(\Hproj(f,[x_n,\ldots,x_{i+2}])(\va,x_{i+1}), \Hproj(f,i+2)(\va,x_{i+1})).$$
We have
$$(\va,\beta)\in P_{i+1}\cap (S'\backslash\zero(\Hproj(f,i+2)))\subseteq P_{i+1}\cap \pi_{i+1}^n(U)\neq\emptyset,$$
and the induction is completed.
\end{proof}

\begin{ex}\label{ex:1-2}
We illustrate the main steps of computing ${\tt OpenSP}(\overline{\Hproj}(f,i),{\Hproj^{\vartriangle}}(f,i),OS)$
using the polynomial $f$ from Examples \ref{ex:1} and~\ref{ex:1-1}. 
\medskip
\begin{description}[leftmargin=3em,style=nextline,itemsep=0.5em]
\item[\sf In:]   $f=(x_3^2+x_2^2+x_1^2-1)(4x_3+3x_2+2x_1-1)\in \ZZ[x_1, x_2,x_3]$
\item[]          $S_{\Hproj(f,[x_3,x_2]),\Hproj(f,2)}=\{-2,-\frac{27}{32},0,\frac{63}{64},2\}$ in $\RR$.

\item[\sf 1:] $P_1 :=\{-2,-\frac{27}{32},0,\frac{63}{64},2\}$
 \;\;\;\;\;\;\;\;($P_1$ has 5 elements, $\va_1,\ldots, \va_{5}$)

\item[\sf 3:] $P_2 := \emptyset$
\item[\sf 5:] $P_2:=P_2\bigcup (\va_1\boxplus {\tt SPOne}(\Hproj(f,[x_3])(\va_1,x_2),\Hproj(f,3)(\va_1,x_2)))$
\item[]       $P_2:=P_2\bigcup (\va_2\boxplus {\tt SPOne}(\Hproj(f,[x_3])(\va_2,x_2),\Hproj(f,3)(\va_2,x_2)))$
\item[]       $\vdots$
\item[]       $P_2:=P_2\bigcup (\va_5\boxplus {\tt SPOne}(\Hproj(f,[x_3])(\va_5,x_2),\Hproj(f,3)(\va_5,x_2)))$
\item[]       $P_2$ now has 13 elements, $\va_1,\ldots, \va_{13}$
\item[\sf 3:] $P_3 := \emptyset$
\item[\sf 5:] $P_3:=P_3\bigcup (\va_1\boxplus {\tt SPOne}(f(\va_1,x_3),\Hproj(f,4)(\va_1,x_3)))$
\item[]       $P_3:=P_3\bigcup (\va_2\boxplus {\tt SPOne}(f(\va_2,x_3),\Hproj(f,4)(\va_2,x_3)))$
\item[]       $\vdots$
\item[]       $P_3:=P_3\bigcup (\va_{13}\boxplus {\tt SPOne}(f(\va_{13},x_3),\Hproj(f,4)(\va_{13},x_3)))$
\item[]       $P_3$ has 36 elements, $\va_1,\ldots, \va_{36}$
\item[\sf Out]       $P_3$
\end{description}
\end{ex}

\begin{rem}\label{re:a1}
As an application of Theorem \ref{thm:reducedcad}, we could design a CAD-like method to get an open sample defined by $f(\xx_n)$ for a given polynomial $f(\xx_n)$. Roughly speaking, if we have already got an open sample defined by $\Hproj(f,[x_n,\ldots,x_j])$ in $\RR^{j-1}$,
according to Theorem \ref{thm:reducedcad}, we could obtain an open sample defined by $f$ in $\RR^n$.
That process could be done recursively.

In the definition of $\Hproj$, we first choose $m$ variables from $\{x_1,\ldots,x_n\}$, compute all projection polynomials under all possible orders of those $m$ variables, and then compute the gcd of all those projection polynomials.
Therefore, Theorem \ref{thm:reducedcad} provides us many ways for designing various algorithms for computing open samples. For example, we may set $m=2$ and choose $[x_n,x_{n-1}]$, $[x_{n-2},x_{n-3}]$, etc. successively in each step. Because there are only two different orders for two variables, we compute the gcd of two projection polynomials under the two orders in each step. Algorithm \ref{TwoHp} is based on this choice.
\end{rem}

\begin{algorithm}[ht]
        \caption{\TwoHp} \label{TwoHp}
        \begin{algorithmic}[1]
                \Require{A polynomial $f \in \ZZ[\xx_n]$ of level $n$.}
                \Ensure{An open sample defined by $f$, {\it i.e.}, a set of sample points which contains at least one point from each connected component of $f\neq0$ in $\RR^{n}$}
        \State $g:=f$;~$L_1:=\{f\}$;~$L_2:=\{1\}$;~$h:=1$;
\While{$i\ge 3$}
        \State $L_1:=L_1\bigcup \{\Hproj(g,[x_i]),\Hproj(g,[x_i,x_{i-1}])\}$; 
        \State $h:=\lc(h,[x_i])$;
        \State $L_2:=L_2\bigcup \{h\}$; 
        \State $h:=\lc(h,[x_{i-1}]){\Hproj^{\vartriangle}(g,[x_i,x_{i-1}],x_{i-1})}$;
        \State $g:=\Hproj(g,[x_{i},x_{i-1}])$;
        \State $i:=i-2$;
\EndWhile
\If {$i=2$}
        \State $L_1:=L_1\bigcup \{\Hproj(g,[x_i])\}$; 
        \State $h:=\lc(h,[x_i])$;
        \State $L_2:=L_2\bigcup \{h\}$;
        \State $g:=\Hproj(g,[x_{i}])$;
        \EndIf
        \State $S$:=${\tt SPOne}(L_1^{[1]},L_2^{[1]})$;
        \State $C$:= ${\tt OpenSP}(L_1,L_2,S)$;
        \State  \Return $C$.
        \end{algorithmic}
\end{algorithm}
\begin{rem}\label{re:a2}
If $\Hproj(f,[x_{n},x_{n-1})]\neq \Bproj(f,[x_{n},x_{n-1}])$ and $n>3$, it is obvious that the scale of projection in Algorithm \ref{TwoHp} is smaller than that of open CAD in Definition \ref{def:opencad}.
\end{rem}

\begin{rem}
It should be mentioned that there are some non-CAD methods for computing sample points in semi-algebraic sets, such as critical point method. For related works, see for example, \cite{basu1998new,safey2003polar,el2007testing,faugere2008classification,Hong_Safey2012}.
\end{rem}

\section{Application: Polynomial Inequality Proving}
\label{sec:improved}
In this section, we combined the idea of $\Hproj$ and the simplified CAD projection operator $\Nproj$ we introduced previously in \cite{han2016proving}, to get a new algorithm for testing semi-definiteness of polynomials.
\begin{defn} \citep{han2016proving}
        Suppose $f\in \ZZ[\xx_{n}]$ is a polynomial of level $n$. Define   
        \begin{align*}
                & {\rm Oc}(f,x_n)=\sqrfree_1(\lc(f,x_{n})), {\rm Od}(f,x_n)=\sqrfree_1(\discrim(f,x_{n})),\\
                & {\rm Ec}(f,x_n)=\sqrfree_2(\lc(f,x_{n})), {\rm Ed}(f,x_n)=\sqrfree_2(\discrim(f,x_{n})),\\
                & {\rm Ocd}(f,x_n)={\rm Oc}(f,x_n)\cup {\rm Od}(f,x_n),\\
                & {\rm Ecd}(f,x_n)={\rm Ec}(f,x_n)\cup {\rm Ed}(f,x_n).
        \end{align*}
        The {\em secondary} and {\em principal parts} of the projection operator $\Nproj$ are defined as
        \begin{align*}
                \Nproj_1(f,[x_{n}])=&{\rm Ocd}(f,x_n),\\
                \Nproj_{2}(f,[x_{n}])=&\{\prod_{g\in {\rm Ecd}(f,x_{n})\setminus {\rm Ocd}(f,x_{n})}{g}\}.
        \end{align*}
        If $L$ is a set of polynomials of level $n$, define
        \begin{align*}
                \Nproj_1(L,[x_{n}])&=\bigcup_{g\in L}{\rm Ocd}(g,x_{n}),\\
                \Nproj_{2}(L,[x_{n}])&=\bigcup_{g\in L}\{\prod_{h\in {\rm Ecd}(g,x_n)\setminus \Nproj_1(L,[x_{n}])}{h}\}.
        \end{align*}
\end{defn}
Based on the projection operator $\Nproj$, we proposed an algorithm, \Proineq, in \citep{han2016proving} for testing semi-definiteness of polynomials. Algorithm \Proineq\ takes a polynomial $f(\xx_n) \in \ZZ[\xx_n]$ as input, and returns whether or not $f(\xx_n) \ge0$ on $\RR^n$.
The readers are referred to \citep{han2016proving} for the details of \Proineq.

The projection operator $\Nproj$ is extended and defined in the next definition.
\begin{defn}
        Let $f\in \ZZ[x_1,\dots,x_n]$ with level $n$. Denote $[\yy]=[y_1,\dots,y_{m}]$, for $1\le m\le n$, where $ y_i \in \{x_1,\dots,x_n\}$ for $1\le i\le m$ and $y_i\neq y_j$ for $i\neq j$. Define
 \[\Nproj(f,[x_i])=\Nproj_2(f,[x_i]),~~ \Nproj(f,[x_i],x_i)=\prod_{g\in\Nproj_1(f,[x_i])}{g}, ~~ \Nproj(f,n)=\Nproj(f,[x_i],x_i).\]
 For $m (m\ge2)$ and $i (1\le i \le m)$, $\Nproj(f,[\yy],y_i)$, $\Nproj(f,[\yy])$ and $\Nproj(f,i)$ are defined recursively as follows.
        \begin{align*}
&\Nproj(f,[\yy],y_i)=\Bproj(\Nproj(f,\hat{[\yy]}_i),y_i), \\
&\Nproj(f,[\yy])=\gcd(\Nproj(f,[\yy],y_1),\ldots,\Nproj(f,[\yy],y_m)),
\end{align*}
where $\hat{[\yy]}_i=[y_1,\ldots,y_{i-1},y_{i+1},\ldots,y_m]$. Define
$$\Nproj(f,i)=\lc(\Nproj(f,i+1),x_i)\frac{\Nproj(f,[x_n,\ldots,x_i],x_i)}{\Nproj(f,[x_n,\ldots,x_i])},$$
\[\overline{\Nproj}(f,i)=\{f,\Nproj(f,[x_{n}]), \ldots, \Nproj(f,[x_{n},\ldots,x_i])\}\] and
\[\widetilde{\Nproj}(f,i)=\{f,\Nproj(f,n),\ldots,\Nproj(f,i)\}.\]
\end{defn}
\begin{thm}\label{thm:2} \citep{han2016proving}
        Given a positive integer $n\ge2$. Let $f\in \ZZ[\xx_n]$ be a non-zero squarefree polynomial and $U$ a connected component of $\Nproj(f,[x_{n}])\neq0$ in $\RR^{n-1}$. If the polynomials in $\Nproj_{1}(f,[x_n])$ are semi-definite on $U$, then $f$ is delineable on $V=U\backslash \bigcup_{h\in \Nproj_{1}(f,[x_n])}\zero(h)$.
\end{thm}
\begin{lem} \label{prop:han2016proving} \citep{han2016proving}
        Given a positive integer $n\ge2$. Let $f\in\ZZ[\xx_{n}]$ be a squarefree polynomial with level $n$ and $U$ a connected open set of $\Nproj(f,[x_{n}])\neq0$ in $\RR^{n-1}$. If $f(\xx_n)$ is semi-definite on $U\times \RR$, then the polynomials in $\Nproj_{1}(f,[x_n])$ are all semi-definite on $U$.
\end{lem}
Now, we can rewrite Theorem \ref{thm:2} in another way. 
\begin{prop}\label{pr:6}
        Let $f\in\ZZ[\xx_{n}]$ be a squarefree polynomial with level $n$ and $U$ a connected component of $\Nproj(f,[x_{n}])\neq0$ in $\RR^{n-1}$. If the polynomials in $\Nproj_{1}(f,[x_n])$ are semi-definite on $U$, then $f$ is weak open delineable on $U\backslash\zero({\Nproj}(f,n))$.
\end{prop}
The proof of Theorem \ref{thm:openweakcad} only uses Theorem \ref{thm:weaktran}, Theorem \ref{thm:weakgcd}, and the fact that $f$ is open weak delineable over $\Bproj(f,[x_n])$ w.r.t. $\{1\}$. Using Proposition \ref{pr:6}, we can prove the following theorem by the same way of proving Theorem \ref{thm:openweakcad}.
\begin{thm} \label{thm:nprojopen}
        Let $j$ be an integer and $2\le j\le n$, $f(\xx_n)\in \ZZ[\bm{x}_n]$, $U$ be any open connected set of $\Nproj(f,[x_n,\ldots,x_j])\neq0$ in $\RR^{j-1}$. If the polynomials in $\bigcup_{i=0}^{n-j} \Nproj_{1}(f,[x_{n-i}])$ are all semi-definite on $U\times \RR^{n-j}$, $f(\xx_n)$ is weak open delineable on $S=U\backslash\zero(\Nproj(f,j))$.
\end{thm}

Theorem \ref{thm:nprojopen} and Proposition \ref{prop:han2016proving} provide us a new way to decide the non-negativity of a polynomial as stated in the next theorem.

\begin{thm}\label{th:6}
Let $f\in\ZZ[\xx_{n}]$ be a squarefree polynomial with level $n$ and $U$ a connected open set of $\Nproj(f,[x_{n},\ldots,x_j])\neq0$ in $\RR^{j-1}$. Denote $S=U\backslash\zero(\Nproj(f,j))$. The necessary and sufficient condition for $f(\xx_n)$ to be positive semi-definite on $U\times \RR^{n-j+1}$ is the following two conditions hold.\\
        $(1)$The polynomials in $\bigcup_{i=0}^{n-j} \Nproj_{1}(f,[x_{n-i}])$ are all semi-definite on $U\times \RR^{n-j}$.\\
        $(2)$There exists a point $\va \in S$ such that $f(\va,x_j,\ldots,x_n)$ is positive semi-definite on $\RR^{n-j+1}$.
\end{thm}
Based on the above theorems, it is easy to design some different algorithms (depending on the choice of $j$) to prove polynomial inequality. For example, the algorithm \TwoPro\ for deciding whether a polynomial is positive semi-definite, which we will introduce later, is based on Theorem \ref{th:6} when $j=n-1$ (Proposition \ref{prop:nproj2}).
\begin{prop}\label{prop:nproj2}
        Given a positive integer $n\ge3$. Let $f\in\ZZ[\xx_{n}]$ be a squarefree polynomial with level $n$ and $U$ a connected open set of $\Nproj(f,[x_{n},x_{n-1}])\neq0$ in $\RR^{n-2}$. Denote $S=U\backslash \zero(\Nproj(f,n-1))$.\\ The necessary and sufficient condition for $f(\xx_n)$ to be positive semi-definite on $U\times \RR^2$ is the following two conditions hold.\\
        $(1)$The polynomials in either $\Nproj_{1}(f,[x_n])$ or $\Nproj_{1}(f,[x_{n-1}])$ are semi-definite on $U\times \RR$.\\
        $(2)$There exists a point $\va\in S$ such that $f(\va,x_{n-1},x_n)$ is positive semi-definite on $\RR^2$.
\end{prop}
\begin{algorithm}[ht]
        \caption{\TwoPro} \label{TwoPro}
        \begin{algorithmic}[1]
                \Require{An irreducible polynomial $f \in \ZZ[\xx_n]$.}
                \Ensure{Whether or not $\forall \va_n\in \RR^n$, $f(\va_n)\ge0.$}
        \If {$n\le2$}
        \If {$\Proineq(f(x_n))$=\textbf{false}}
        \Return  \textbf{false}
        \EndIf
        \Else
                \State $L_1:=\Nproj_{1}(f,[x_n])\bigcup \Nproj_{1}(f,[x_{n-1}])$
                \State $L_2:=\Nproj(f,[x_{n},x_{n-1}])$
                \For {$g$ in $L_1$}
                \If {\TwoPro$(g)=$\textbf{false}}
                \Return  \textbf{false}
                \EndIf
                \EndFor
                \State $C_{n-2}:=$ A reduced open CAD of $L_2$ w.r.t. $[x_{n-2},\ldots,x_2]$, which satisfies that $\zero(\Nproj(f,n-1))\cap C_{n-2}=\emptyset$.
                \If{$\exists \va_{n-2}\in C_{n-2}$ such that $\Proineq(f(\va_{n-2},x_{n-1},x_n))$=\textbf{false}}
                \Return \textbf{false}
        \EndIf
                \EndIf\\
                \Return\textbf{true}
        \end{algorithmic}
\end{algorithm}

\section{Application: Copositive problem}\label{sec:copositive}
\begin{defn}
A real $n\times n$ matrix $A_n$ is said to be {\em copositive} if $\xx_nA_n\xx_n^T\ge0$ for every nonnegative vector $\xx_n$. For convenience, we also say the form $\xx_nA_n\xx_n^T$ is copositive if $A_n$ is copositive.
\end{defn}
The collection of all copositive matrices is a proper cone; it includes as a subset the collection of real positive-definite matrices. For example, $xy$ is copositive but it is not positive semi-definite.

In general, to check whether a given integer square matrix is not copositive, is NP-complete \citep{murty1987some}.
This means that every algorithm that solves the problem, in the worst case, will require at least
an exponential number of operations, unless P=NP. For that reason, it is
still valuable for the existence of so many incomplete algorithms discussing
some special kinds of matrices \citep{parrilo2000structured}. For
small values of $n$ ($\le6$), some necessary and sufficient conditions have been
constructed \citep{hadeler1983copositive,andersson1995criteria}. We refer the reader to \citep{hiriart2010variational} for a more detailed introduction to copositive matrices.

From another viewpoint, this is a typical real quantifier elimination problem, which can be solved by standard tools of real quantifier elimination ({\it e.g.}, using typical CAD). Thus, any CAD based QE algorithm can serve as a complete algorithm for
deciding copositive matrices theoretically. Unfortunately, such algorithm is not efficient in practice since the computational complexity of CAD is double exponential in $n$. 

To test the copositivity of the form $\xx_nA_n\xx_n^T$, is equivalent to test the nonnegativity of the form $(x_1^2,\ldots,x_n^2)A_n(x_1^2,\ldots,x_n^2)^T$.
In this section, we give a singly exponential incomplete algorithm with time complexity $\OO(n^2 4^n)$ based on the new projection operator proposed in the last section and Theorem \ref{th:6}.
We remark here that the results of \cite{basu1998new} allow to solve the problem in time singly exponential in $n$. However, the constants in the exponent are not made explicit. The constants of our bound are explicit and very low.

Let us take an example to illustrate our idea. Let $$F:=ax^4+bx^2y^2+cy^4+dx^2+ey^2+f,$$
be a squarefree polynomial, where $a,b,c\in \ZZ,d,e,f\in\ZZ[\zz_{n}]$ and $a\neq0$, $c\neq 0$.

To test the nonnegativity of $F$, we could apply typical CAD-based methods directly, {\it i.e.}, we can use Brown's projection operator. 
In general, we have
$$\Bproj(F,[x])=(cy^4+ey^2+f)a(4acy^4+4aey^2+4af-b^2y^4-2by^2d-d^2).$$
We then eliminate $y$,
\begin{align*}
  &\Bproj(F,[x,y])=\Bproj(\Bproj(F,[x]),y)\\
  =&fac(4fc-e^2)(d^2c-edb+fb^2)(4af-d^2)(4ac-b^2)(4afc-ae^2-d^2c+edb-fb^2).
\end{align*}
If $d,e$ are polynomials of degree $2$ and $f$ is a polynomial of degree $4$ (copositive problem is in this case), the degree of the polynomial $\Bproj(F,[x,y])$ is 20 while the original problem is of degree 4 only. That could help us understand why typical CAD-based methods do not work for copositive problems with more than 5 variables in practice.

Now, we apply our new projection operator. Notice that
$$\Res(F,\frac{\partial F }{\partial x}, x)=16(cy^4+ey^2+f)F_1^2,$$
where $F_1=a(4acy^4+4aey^2+4af-b^2y^4-2by^2d-d^2)$.

If $cy^4+ey^2+f$ and $F_1$ are nonzero and squarefree, $\Nproj(F,[x])=F_1$. Thus, in order to test the nonnegativity of $F$, we only need to test the semi-definitness of $cy^4+ey^2+f$, choose sample points defined by
$\Nproj(F,[x])\neq0$ (we also require that $cy^4+ey^2+f$ does not vanish at those sample points) and test the nonnegativity of $F$ at these sample points.

On the other side,
$$\Res(F, \frac{\partial F}{\partial y},y)=16(ax^4+dx^2+f)F_2^2,$$
where $F_2=c(4cx^4a+4cdx^2+4fc-b^2x^4-2bx^2e-e^2)$.

Similarly, if $ax^4+dx^2+f$ and $F_2$ are nonzero and squarefree, $\Nproj(F,[y])=F_2$. In order to test the nonnegativity of $F$, we only need to test the semi-definitness of $ax^4+dx^2+f$, choose sample points defined by
$\Nproj(F,[y])\neq0$ (we also require that $ax^4+dx^2+f$ does not vanish at those sample points) and test the nonnegativity of $F$ at these sample points.

Under some ``generic" conditions ({\it i.e.}, some polynomials are nonzero and squarefree), 
we only need to test the semi-definitness of $ax^4+dx^2+f$ and $cy^4+ey^2+f$, choose sample points $T_2$ defined by
$\Nproj(F,[x,y])=\gcd(\Bproj(\Nproj(F,[x]),y)$, $\Bproj(\Nproj(F,[y]),x))=(4ac-b^2)(4afc-ae^2-d^2c+edb-fb^2)\neq0$ (we also require that $\Res(\Nproj(F,[x]),y)$ does not vanish at $T_2$), obtain sample points $T_1$ defined by
 $\Nproj(F,[x])\neq0$ from $T_2$ (we also require that $cy^4+ey^2+f$ does not vanish at $T_1$)
 and test the nonnegativity of $F$ at $T_1$.

Again, if $d,e$ are polynomials of degree $2$ and $f$ is a polynomial of degree $4$, both the degree of $\Nproj(F,[x])$ and $\Nproj(F,[x,y])$ are exactly $4$. It indicates that our new projection operator may control the degrees of polynomials in projection sets. Moreover, we point out that
\begin{align*}
  \Nproj(F,[x,y])&=(4afc-ae^2-d^2c+edb-fb^2)\\
  &=4 \det\left( \begin{bmatrix} a & \frac{b}{2} &\frac{d}{2}\\ \frac{b}{2} & c &\frac{e}{2}\\
      \frac{d}{2}&\frac{e}{2}&f\end{bmatrix}\right).
\end{align*}
Before giving the result, we introduce some new notations and lemmas for convenience.
\begin{defn}[Sub-sequence]
  An array $I$ is called a {\em sub-sequence} of sequence $\{1,\dots,n\}$ if for any $i$-th component of $I$, $I[i]\in \{1,\dots, n\} $ and $I[i]< I[i+1]$ for $i=1,\dots,|I|-1$. 
\end{defn}

For a sub-sequence $I$ of $\{1,\dots, n\}$ with $m=|I|$, denote $\xx_I=[x_{I[1]},\ldots,x_{I[m]}]$, $\overline{\xx_I}=\{x_i\mid i\not\in I\}$ and
$A_I=(a_{{I[i],I[j]}})_{1\le i,j\le m}$ the sub-matrix of $A_{n}$.

Let \begin{equation}\label{eq:f}
      f(\xx_n) =\sum_{1\le i,j\le n} a_{i,j}x_ix_j+\sum_{i=1}^n (a_{i,n+1}+a_{n+1,i})x_i+a_{n+1,n+1}=(\xx_n,1)A_{n+1}(\xx_n,1)^T
    \end{equation}
be a quadratic polynomial in $\xx_n$, where $A_{n+1}=(a_{i,j})_{(n+1)\times(n+1)}$, 
$a_{i,j}=a_{j,i}\in \ZZ[\zz_s]$ for $1\le i,j\le n+1$. Set $F(\xx_n)=f(x_1^2,\ldots,x_n^2)$. It is not hard to see that (please refer to the proof of Theorem \ref{thm:cop}), for a given sub-sequence $I$ of $\{1,\dots, n\}$ with length $m$,  there exist some polynomials $p_1,\ldots,p_{m+1}\in \ZZ[\zz_s,\overline{\xx_I}]$ such that $F(\xx_n)=(\xx_I^2,1)P_I(\xx_I^2,1)^T$ where
  $$P_I=\begin{bmatrix} A_I & (p_{1}, \cdots , p_{m})^T \\ (p_{1}, \cdots , p_{m}) & p_{m+1} \end{bmatrix}.$$
For convenience, we denote $P_{[1,\dots,m]}$ by $P_{m+1}(x_{m+1},\ldots,x_n)$ or simply $P_{m+1}$. In particular, $P_{n+1}=A_{n+1}$ and $F(\xx_n)=(x_1^2,\ldots, x_i^2,1)P_{i+1}(x_1^2,\ldots,x_i^2,1)^T$.
\begin{ex}
  Suppose $F(x_1,x_2,x_3)=x_1^4+2x_2^4+4x_1^2x_2^2-2x_1^2x_3^2+4x_2^2x_3^2+8x_1^2z^2+5x_3^4+z^4$. 
  If $I=[1,2]$, then $\xx_I=[x_1,x_2],\overline{\xx_I}=\{x_3\}$,
  $$P_I=\begin{bmatrix} 1 & 2 & 4z^2-x_3^2\\
    2 & 2 &        2x_3^2 \\
    4z^2-x_3^2 & 2x_3^2 & 5x_3^4+ z^4  \end{bmatrix}.$$
  If $I=[1]$, then
  $\xx_I=[x_1],\overline{\xx_I}=\{x_2,x_3\}$,
  $$P_I=\begin{bmatrix} 1 &\   2x_2^2+4z^2-x_3^2\\
    2x_2^2+4z^2-x_3^2 & \  2x_2^4+5x_3^4+ z^4  \end{bmatrix}.$$
\end{ex}

\begin{lem}\label{lem:matrix}
  Suppose $R$ is a square matrix with order $n$, $P$ is an invertible  square matrix with order $k<n$, $Q\in \RR^{k\times (n-k)}$, $M\in \RR^{(n-k)\times k }$ and $N\in \RR^{(n-k)\times (n-k) }$. If $R$ can be written as partitioned matrix
  $$R=\begin{bmatrix} P & Q \\ M & N \end{bmatrix},$$ then
  $$\det(R)=\det(P)\det(N-MP^{-1}Q).$$
\end{lem}
\begin{pf}
  It is a well known result in linear algebra.
\end{pf}
For a square matrix $M$, we use $M^{(i,j)}$ to denote the determinant of the sub-matrix obtained by deleting the $i$-th row and the $j$-th column of $M$.
\begin{thm}
  \label{thm:cop}
Suppose $f\in\ZZ[\zz_s][\xx_n]$ is defined as in (\ref{eq:f}). Set $F(\xx_n)=f(x_1^2,\ldots,x_n^2)$. If
  \begin{enumerate}
  \item $F(\xx_n)|_{x_i=0} $ is nonzero and squarefree for any $i\in \{1,\dots,n\}$;
  \item $\det(A_I)=P_{I}^{(|I|+1,|I|+1)}$ is nonzero and squarefree for any sub-sequence $I$ of $\{1,\ldots,n \}$, and $\gcd(A_{I}^{(1,1)},\dots,A_{I}^{(|I|,|I|)})=1$ for any any sub-sequence $|I|\ge2$ of $\{1,\ldots,n \}$;
  \item  $\gcd(P_{I}^{(1,1)},\dots,P_{I}^{(|I|,|I|)})=1$ for any sub-sequence $|I|\ge2$ of $\{1,\ldots,n \} $;
  \item $\det(P_{I})$ is nonzero and squarefree for any sub-sequence $I$ of $\{1,\ldots,n \}$;
  \item $\gcd(\det(P_{I}),\det(A_{I}))=1$ for any sub-sequence $I$ of $\{1,\ldots,n\}$,
  \end{enumerate}
  then $\Nproj(F,[\xx_n])=\det(A_{n})\det(A_{n+1})$. 
\end{thm}

\begin{proof}
  We prove the theorem by induction on $n$.

  If $n=1$, $F(\xx_1)=a_{1,1}x_1^4+2a_{1,2}x_1^2+a_{2,2}$. Then $$\Res(F,\frac{\partial F}{\partial x_1},x_1)=256a_{1,1}^2a_{2,2}(a_{1,1}a_{2,2}-a_{1,2}^2)^2.$$ By conditions (1), (2), (4) and (5), $a_{1,1}\neq0
  , a_{2,2}\neq0, a_{1,1}a_{2,2}-a_{1,2}^2\neq0$, and $a_{2,2}$ and $a_{1,1}a_{2,2}-a_{1,2}^2$ are two coprime squarefree polynomials. Thus,
  $\Nproj(F,[\xx_1])=a_{1,1}(a_{1,1}a_{2,2}-a_{1,2}^2)=\det(A_1)\det(A_2)$.

  Assume that the conclusion holds for any quadratic polynomials with $k$ variables where $1\le k < n$. When $n=k$,
  let $I$ be a sub-sequence of $\{1,\dots,n\}$ with $|I|=n-1$. Without loss of generality, we assume that $I=[1,\dots,n-1]$. Set $A_{I}=(a_{i,j})_{1\le i,j< n}$, $B=(a_{1,n},\ldots,a_{n-1,n})$, $C=(a_{1,n+1},\ldots,a_{n-1,n+1})$, and $D=a_{n,n}x_n^4+2a_{n,n+1}x_n^2+a_{n+1,n+1}$. Then, $F(\xx_n)$ could be written as
  \begin{align*}
    F(\xx_{n})=&\sum_{1\le i,j< n} a_{i,j}x_i^2x_j^2+\sum_{i=1}^{n-1} (a_{i,n}+a_{n,i})x_i^2x_n^2+(a_{n,n}x_n^4+2a_{n,n+1}x_n^2+a_{n+1,n+1})\\
    =&(x_1^2,\dots,x_{n-1}^2,1)P_{I}(x_1^2,\dots,x_{n-1}^2,1)^T,
  \end{align*}
  where  $$P_{I}=\begin{bmatrix} A_{I} & (Bx_n^2+C)^T \\ Bx_n^2+C & D \end{bmatrix}.$$
  By assumption, $F|_{x_i=0}$ is squarefree for $i\in I$
  , and $\det(A_I')$ is nonzero and squarefree, $\gcd(A_{I'}^{(1,1)},\dots,A_{I'}^{(|I'|,|I'|)})=1$, $\gcd(P_{I'}^{(1,1)},\dots,P_{I'}^{(|I'|,|I'|)})=1$, $\gcd(\det(P_{I'}),\det(A_{I'}))=1$ for any sub-sequence $I'$ of $I$ with $|I'|\le n-1 $. Thus, by induction hypothesis, $\Nproj(F(\xx_n),[\xx_I])=\det(A_I)\det(P_{I})$.

  In the following, we compute $\det(P_{I})$. By assumption, $\det(A_{I})=P_{I}^{(n,n)}$ is nonzero and squarefree. According to Lemma \ref{lem:matrix},
  \begin{align*}
    \det(P_{I})=&\det(A_{I})(D-(Bx_n^2+C)A_I^{-1}(Bx_n^2+C)^T)\\
    =&\det(A_{I})\left((a_{n,n}-BA_I^{-1}B^T)x_n^4+2(a_{n,n+1}- 
    BA_I^{-1}C^T)x_n^2+a_{n+1,n+1}-CA_I^{-1}C^T\right)\\
    =&\det(A_I)(\lambda x_n^4+2\mu x_n^2+\nu),
  \end{align*}
  where $\lambda=(a_{n,n}-BA_I^{-1}B^T)$, $\mu=a_{n,n+1}-BA_I^{-1}C^T$, $\nu=a_{n+1,n+1}-CA_I^{-1}C^T$.
  By Lemma \ref{lem:matrix} again,
  \begin{align*}
    \det(A_I)\lambda=&\det(A_I)(a_{n,n}-BA_I^{-1}B^T)\\
    =&\det\left(\begin{bmatrix} A_I&B^T\\
        B&a_{n,n}\end{bmatrix}\right) \\
    =& A_{n+1}^{(n+1,n+1)},\numberthis \label{eq:1} \\
    \det(A_I)\nu=&\det(A_I)(a_{n+1,n+1}-CA_I^{-1}C^T)\\
    =&\det\left(\begin{bmatrix} A_I&C^T\\
        C&a_{n+1,n+1}\end{bmatrix}\right)\\
    =& A_{n+1}^{(n,n)}, \numberthis \label{eq:2}
  \end{align*}

  Thus, both $\det(A_I)\lambda$ and $\det(A_I)\nu$ are the determinants of some principal sub-matrices of $A_{n+1}$ with order $n$, respectively.

  Let $H=\det(P_{I})$, according to Lemma \ref{lem:matrix}, it is clear that
  \begin{align*}
    \Res(H,\frac{\partial H }{\partial x_n }, x_n)=&256\det(A_I)^7\lambda^2\nu(\mu^2-\lambda\nu)^2\\
    =&256 \det(A_I)^7\lambda^2\nu \det\left(\begin{bmatrix} \lambda & \mu \\
        \mu & \nu \end{bmatrix}\right)^2\\
    =& 256 \det(A_I)^7\lambda^2\nu \det \left(\begin{bmatrix} a_{n,n} & a_{n,n+1} \\
        a_{n+1,n} & a_{n+1,n+1}\end{bmatrix}\right. 
    -\left . \left(\! \begin{array}{c}
          B \\
          C
        \end{array}
        \!\right)
      A_I^{-1}\left(\! \begin{array}{c}
          B \\
          C
        \end{array}
        \!\right)^T\right)^2\\
    =& 256 \det(A_I)^5\lambda^2\nu \det\left(\begin{bmatrix} A_I&B^T&C^T\\
        B& a_{n,n}&a_{n,n+1}\\
        C&a_{n+1,n},&a_{n+1,n+1}\end{bmatrix}\right)^2\\
    =&256  \det(A_I)^5\lambda^2\nu  \det(A_{n+1})^2\\
    =& 256 \det(A_I)^2  \left(A_{n+1}^{(n+1,n+1)}\right)^2 A_{n+1}^{(n,n)} \det(A_{n+1})^2 \\
    =& 256 (A_{n}^{(n,n)})^2 \det(A_{n})^2 A_{n+1}^{(n,n)} \det(A_{n+1})^2.\numberthis \label{eq:3}
  \end{align*}

  Since $\det(A_I)$ and 
  $A_{n+1}^{(n+1,n+1)}$ are nonzero, according to (\ref{eq:3}), we have

  \begin{align*}
    \Nproj(F,[\xx_n],n)    =&\sqrfree(A_{n+1}^{(n,n)}A_{n}^{(n,n)}\det(A_{n}) \det(A_{n+1})). 
  \end{align*}
Similarly, for $1\le i\le n,$ we have
  \begin{align*}
    \Nproj(F,[\xx_n],i)=&\sqrfree( A_{n+1}^{(i,i)}A_{n}^{(i,i)}\det(A_{n})\det(A_{n+1})). \numberthis \label{eq:4}
  \end{align*}
By assumption, $\gcd(A_{n+1}^{(1,1)},\cdots, A_{n+1}^{(n,n)})=1$, $\gcd(A_{n}^{(1,1)},\cdots, A_{n}^{(n,n)})=1$, and $\det(A_{n})$, $\det(A_{n+1})$ are two nonzero squarefree polynomials with $\gcd(\det(A_{n}),\det(A_{n+1}))=1$, thus
  $$\Nproj(F,[\xx_n])=\gcd(\Nproj(F,[\xx_n],1 ),\ldots,\Nproj(F,[\xx_n],n))=\det(A_{n})\det(A_{n+1}).$$
  That completes the proof.
\end{proof}

\begin{thm}\label{thm:gcop}
If the coefficients $a_{i,j}$ of $f(\xx_{n})$ in Theorem \ref{thm:cop} are pairwise different real parameters except that $a_{i,j}=a_{j,i}$, and $F(\xx_n)=f(x_1^2,\ldots,x_n^2)$. Then all the four hypotheses (1)-(5) of Theorem \ref{thm:cop} hold. As a result, $\Nproj(F,[\xx_n])=\det(A_{n})\det(A_{n+1})$. 
\end{thm}
\begin{proof}
  It is clear that the hypotheses (1) and (2) of Theorem \ref{thm:cop} hold. We claim that that for any given $m$, $|I|=m$, $P_{I}^{(i,i)}$ and $\det(P_{I})$ are pairwise nonconstant different irreducible polynomials in $\ZZ[\bm{a}_{i,j}][\xx_n]$ for all $n\ge m$, so the hypotheses (3),(4) and (5) of Theorem \ref{thm:cop} also hold. Here we denote $\bm{a}_{i,j}=(a_{1,1},\ldots,a_{n+1,n+1})$.

  We only prove that $\det(P_{I})$ is a nonconstant irreducible polynomial. The other statements of the claim can be proved similarly.

  We prove the claim by introduction on $n$. If $n=m$, it is clear that the claim is true. Assume that the theorem holds for integers $m\le l\le n-1$. We now consider the case $l=n$. Without of loss of generality, we assume that $I=[1,2,\ldots,m]$.

Recall that
 $F(\xx_n)$ could be written as
        \begin{align*}
                F(\xx_n)=(\xx_m^2,1)P_{I,F}(\xx_m^2,1)^T,
        \end{align*}
 where
 $$P_{I,F}=\begin{bmatrix} A_{I} & (p_{1,n}, \cdots , p_{m,n})^T \\ (p_{1,n}, \cdots , p_{m,n}) & p_{m+1,n} \end{bmatrix},$$
  and $p_{i,n}=\sum_{j=m+1}^na_{i,j}x_j^2+a_{i,n+1}$ for $1\le i\le m$, and $p_{m+1,n}=\sum_{j=m+1}^n a_{j,j}x_j^4+\sum_{j=m+1}^n 2a_{n+1,j}x_j^2+\sum_{m+1\le i<j\le n} 2a_{i,j}x_i^2x_j^2+a_{n+1,n+1}$.
  Let $B=(a_{1,n},\ldots,a_{m,n})$, $C=(p_{1,n-1},\ldots,p_{m,n-1})$, and $D=p_{m+1,n}=a_{n,n}x_n^4+2a_{n,n+1}x_n^2+2\sum_{m+1\le i< n} a_{i,n}x_i^2x_n^2+a'_{n+1,n+1}$,
  where $a'_{n+1,n+1}$ is a polynomial with $\deg(a'_{n+1,n+1},x_n)=0$.
Now $P_{I,F}$ can be simply written as
  $$P_{I,F}=\begin{bmatrix} A_{I} & (Bx_n^2+C)^T \\ Bx_n^2+C & D \end{bmatrix}.$$
  In the following, we compute $\det(P_{I,F})$. By Lemma \ref{lem:matrix},
   \begin{align*}
    \det(P_{I,F})=&\det(A_{I})(D-(Bx_n^2+C)A_I^{-1}(Bx_n^2+C)^T)\\
    =&\det(A_{I})\left((a_{n,n}-BA_I^{-1}B^T)x_n^4+\right.2(a_{n,n+1}+\sum_{m+1\le i< n} a_{i,n}x_i^2-BA_I^{-1}C^T)x_n^2\\
    & \left. +a'_{n+1,n+1}-CA_I^{-1}C^T\right)\\
    =&\det(A_I)(\lambda x_n^4+2\mu x_n^2+\nu),
  \end{align*}
  where $\lambda=(a_{n,n}-BA_I^{-1}B^T)$, $\mu=a_{n,n+1}+\sum_{m+1\le i< n} a_{i,n}x_i^2-BA_I^{-1}C^T$, $\nu=a'_{n+1,n+1}-CA_I^{-1}C^T$.
  By Lemma \ref{lem:matrix},
  \begin{align*}
    \det(A_I)\lambda=&\det(A_I)(a_{n,n}-BA_I^{-1}B^T)\\
    =&\det\left(\begin{bmatrix} A_I&B^T\\
        B&a_{n,n}\end{bmatrix}\right) \\
    \det(A_I)\nu=&\det(A_I)(a'_{n+1,n+1}-CA_I^{-1}C^T)\\
    =&\det\left(\begin{bmatrix} A_I&C^T\\
        C&a'_{n+1,n+1}\end{bmatrix}\right).\\
  \end{align*}
 We have $ \det(A_I)\lambda=\det(P_{I,G})$, $\det(A_I)\nu=\det(P_{I,H})$, where
  $$G(\xx_m)=(\xx_m,1)\left(\begin{bmatrix} A_I&B^T\\
        B&a_{n,n}\end{bmatrix}\right)(\xx_m,1)^T,$$ and
  $$H(\xx_{n-1})=(\xx_m,1)\left(\begin{bmatrix} A_I&C^T\\
        C&a'_{n+1,n+1}\end{bmatrix}\right)(\xx_m,1)^T. $$
 By induction, $\det(A_I)\lambda$ and $\det(A_I)\nu$ are two different non-constant irreducible polynomials. Since $\deg(\det(A_I)\mu,a_{n,n+1})>0$, and $\deg(\det(A_I)\lambda,a_{n,n+1})=\deg(\det(A_I)\nu,a_{n,n+1})=0$, it is clear that $\det(A_I)\mu\neq \pm (\det(A_I)\lambda\cdot \det(A_I)\nu+1)$, $\det(A_I)\mu\neq \pm (\det(A_I)\lambda+\det(A_I)\nu)$.
 Now the result follows from Lemma \ref{lem:irr}. We are done.
   \end{proof}


  \begin{lem} \label{lem:irr}
   Let $\mathcal{R}$ be a UFD with units $\pm1$. Let $a,b,c\in \mathcal{R}$, where $b\neq \pm (ac+1)$, $b\neq \pm (a+c)$, and $a,c$ are two non-unit coprime irreducible elements in $\mathcal{R}$, then $T(x)=ax^4+bx^2+c$ is an irreducible polynomial in $\mathcal{R}[x]$. 
   \begin{proof}
   Otherwise, we may assume $T(x)=g(x)h(x)$, where $g,h$ are two nonconstant polynomials in $\mathcal{R}[x]$.
   Notice that if $\alpha\in \mathcal{R}$ is a root of $T(x)$, then $-\alpha$ is also a root of $T(x)$, thus $(x^2-\alpha^2)$ is a factor of $T$. Thus, we may assume that $\deg(g)=\deg(h)=2$. Let $g=g_0+g_1x+g_2x^2$, $h=h_0+h_1x+h_2x^2$, where $g_i,h_i\in\mathcal{R}$.
   By comparing the coefficients of $T$ with $gh$, we have $c=g_0h_0$, $0=g_0h_1+g_1h_0$, $h_1g_2+h_2g_1=0$. Assume that $c|g_0$, then $c\nmid h_0$. if $h_1\neq0$, let $l$ be the largest integer such that $c^l|h_1$, then $l+1$ is the largest integer such that $c^{l+1}|g_1$. But $h_1g_2+h_2g_1=0$, so $c|g_2$, and $c|gh$, which contradicts with $(a,c)=1$. We must have $h_1=0$, and $g_1=0$.
   We assume that $g_0=c,h_0=1$. Now, there are four cases $(g_2=\pm a,h_2=\mp1)$ or $(g_2=\pm1,h_2=\mp a)$. All the four cases will contradict with the assumption that $b\neq \pm (ac+1)$, $b\neq \pm (a+c)$.
   \end{proof}
  \end{lem}


\begin{thm}
Suppose $g(\xx_{n})=\sum_{1\le i,j\le n} a_{i,j}x_ix_j=\xx_{n}A_{n}\xx_{n}^T$ is a quadratic polynomial
        where $a_{i,j}$ are pairwise different real parameters except that $a_{i,j}=a_{j,i} (1\le i,j\le n)$. Let $G(\xx_n)=g(x_1^2,\ldots,x_n^2)$, 
        then $\Nproj(G,[\xx_n])=\det(A_{n})$. 
\end{thm}
\begin{proof}
By Theorem \ref{thm:gcop}, we have
$$\Nproj(G,[\xx_n],n)=\Bproj(\det(A_{n-1})\det(A_n)x_n,x_{n})=\det(A_{n-1})\det(A_n).$$
Therefore,
        $\Nproj(G,[\xx_n])=\gcd(\Nproj(G,[\xx_n],1 ),\ldots,\Nproj(G,[\xx_n],n))=\det(A_{n}).$
\end{proof}
By similar method, we can prove that
\begin{thm}
Suppose $g(\xx_{n})=\sum_{1\le i,j\le n} a_{i,j}x_ix_j=\xx_{n}A_{n}\xx_{n}^T$ is a quadratic form
        where $a_{i,j}$ are pairwise different real parameters except that $a_{i,j}=a_{j,i} (1\le i,j\le n)$. Let $A_{n}=(a_{i,j})_{i,j=1}^{n}$. 
        Then $\Hproj(g,[\xx_n])=\det(A_{n})=\discrim(g,[\xx_n])$, where $\discrim(g,[\xx_n])$ is the discriminant of the quadratic form $g$, it is an irreducible polynomial in $\ZZ[\bm{a}_{i,j}]$.
\end{thm}
This theorem implies that, for a class of polynomial $g$, $\Hproj(g,[\xx_n])$ may coincide with its discriminant.

Let $g$ be a ``generic'' form in $n$ variables with degree $d$,
$$g(\xx_n,\bm{C}_{\bm{\alpha}})=\sum_{|\bm{\alpha}|=d}C_{\bm{\alpha}}\bm {x}^{\bm{\alpha}},$$
where $\bm{\alpha}=(\alpha_1,\ldots,\alpha_n)$, $|\bm{\alpha}|=\sum_{i=1}^n \alpha_i$, $\bm {x}^{\bm{\alpha}}=\prod_{i=1}^n x_i^{\alpha_i}$, $\{\bm{C}_{\bm{\alpha}}\}=\{C_{\bm{\alpha}}||\bm{\alpha}|=d\}$, and $N=(\begin{subarray}{c}n+d-1\\n-1 \end{subarray})$.

It was proved in \cite{han2016multivariate} that

(1) the multivariate discriminant $\discrim(g,[x_n,\ldots,x_1])$ of the generic form $g(\xx_n)$ with even degree $d$ is an irreducible factor of $\Hproj(g,[\xx_n])$.

(2) for generic form $g(x,y,z)$ in three variables with degree $d$, we have
  $$\Hproj(g,[x,y,z])=\discrim(g,[x,y,z]).$$

We conjecture that for generic form $g(x_1,\ldots,x_n)$, we have
  $$\Hproj(g,[x_n,\ldots,x_1])=\discrim(g,[x_n,\ldots,x_1]).$$

Theorem \ref{thm:cop} and Theorem \ref{thm:gcop} show that, for a generic copositive problem, we can compute the projection set $\overline{\Nproj}(F,{n-1})=\{f,\Nproj(f,[x_{n}]), \ldots, \Nproj(f,[x_{n},\ldots,x_2)]\}$ directly. Based on the theorem, it is easy to design a complete algorithm for solving copositive problems. However, for an input $f(\xx_n)$, checking whether $f(\xx_n)$ satisfies the hypothesis (3) of Theorem \ref{thm:cop} is expensive. Therefore we propose a special incomplete algorithm \TCPT\ for {\bf c}opositive {\bf m}atrix {\bf t}esting, which is formally described as Algorithm \ref{cop}.

\begin{algorithm}[ht]
  \caption{\TCPT}
  \label{cop}
  \begin{algorithmic}[1]
    \Require{An even quartic squarefree polynomial $F(\xx_n) \in \ZZ[\xx_n]$, $n\ge1$, with an ordering $x_n\prec x_{n-1}\cdots \prec x_1$ and a set $Q$ of nonnegative polynomials.}
    \Ensure{Whether or not $F(\xx_n) \ge0$ on $\RR^n$}
    \If {$F\in Q$}
      \Return \textbf{true}
    \EndIf
    \For {$i$ from $1$ to $n$}
    \If {$\TCPT(F(x_1,x_2,\ldots,x_n)|_{x_i=0},Q)=$\textbf{false}}
    \Return  \textbf{false}
    \Else ~ $Q:=Q\cup (F(x_1,x_2,\ldots,x_n)|_{x_i=0})$
    \EndIf
    \EndFor
    \State $g_n(x_n):=\det(P_{n})$ \Comment{Recall that $F(\xx_n)=(x_1^2,\ldots, x_{n-1}^2,1)P_{n}(x_1^2,\ldots,x_{n-1}^2,1)^T$.}
    \State $O:={\tt SPOne}(g_n,1)$
    \For {$i$ from $2$ to $n$}
    \State $S:=\emptyset$
    \State $g_{n-i+1}(x_{n-i+1},\ldots,x_n):=\det(P_{n-i+1})$ \Comment{Recall that $F(\xx_n)=(x_1^2,\ldots,x_{n-i}^2,1)P_{n-i+1}(x_1^2,\ldots,x_{n-i}^2,1)^T$.}
    \For {$\va$ in $O$}
    \State $S:=S\bigcup (\va\boxplus {\tt SPOne}(g_{n-i+1}(x_{n-i+1},\va),1))$
    \EndFor
    \State $O:=S$
    \EndFor
    \If{$\exists \va_{n}\in O$ such that $F(\va_{n})<0$} 
    \Return \textbf{false}
    \EndIf
    \State\textbf{return} \textbf{true}
  \end{algorithmic}
\end{algorithm}

\begin{rem}
In Algorithm \ref{cop}, we do not check the hypotheses of Theorem \ref{thm:cop}. Thus the algorithm is incomplete. However, the algorithm still makes sense because almost all $f(\xx_n)$ defined by Eq. (\ref{eq:f}) satisfy the hypotheses. On the other hand, for an input $f(\xx_n)$, checking whether $f(\xx_n)$ satisfies the hypothesis (3) of Theorem \ref{thm:cop} is expensive but the other three hypotheses are easy to check. Furthermore, $f(\xx_n)$ is degenerate when some hypotheses do not hold and such case can be easily handled. Therefore, when implementing Algorithm \TCPT, we take into account those possible improvements. The details are omitted here.
\end{rem}

{\bf Complexity analysis of  Algorithm \ref{cop}}.
We analyze the upper bound on the number of algebraic operations of Algorithm \ref{cop}.

We first estimate the complexity of computing $\det(P_k)$ for $1\le k\le n$. Because the entries of the last row and the last column of $P_k$ are polynomials with $k^2$ terms and the other entries are integers, we expand $\det(P_k)$ by minors along the last column and then expand the minors again along the last rows. Therefore, the complexity of computing $\det(P_k)$, {\it i.e.}, $g_k$, is $\OO(k^2(k-2)^3+k^2(n-k)^2)$. Since $g_k(x_k,\va)$ is an even quartic univariate polynomial, the complexity of real root isolation for $g_k(x_k,\va)$ is $\OO(1)$ and we only need to choose positive sample points when calling {\tt SPOne}. That means {\tt SPOne}$(g_k(x_k,\va),1)$ returns at most 3 points. Thus the scale of $O$ in line 13 is at most $3^{i-1}$. The cost of computing $g_k(x_k,\va)$ is $\OO(k^2)$ for each sample point $\va$. Then the cost of the ``for loop" at lines 10-17 is bounded by
$$\OO(\sum_{i=1}^{n}(i^2(i-2)^3+i^2(n-i)^2+i^23^{i-1}))=\OO(n^23^n).$$

In line $18$ of Algorithm \ref{cop}, the number of checking $F(\va_n)$ is at most $3^n$. And the complexity of every check in line $18$ is $\OO(n^2)$ since $F$ has at most $n^2$ terms. Then, the complexity of line $18$ is bounded by $\OO(n^23^n)$. Therefore, the complexity of lines $8-19$ is bounded by $\OO(n^23^n)$.

The scale of the set $Q$ is at most $\sum_{k=0}^{n} {n \choose k}=2^n.$ So, the cost of all recursive calls is bounded by
$$\OO(\sum_{k=0}^n{n \choose k}3^kk^2)=\OO(n^22^{2n}).$$

In conclusion, the complexity of Algorithm \ref{cop} is bounded by $\OO(n^22^{2n}).$

\begin{rem}
By a more careful discussion, we may choose at most two sample points on every call {\tt SPOne}$(g_k(x_k,\va),1)$. That will lead to an upper bound complexity, $\OO(3^n n^2)$.
\end{rem}

\section{Examples}
\label{sec:applicat}

The Algorithm \TwoHp, Algorithm \TwoPro, and Algorithm \TCPT\ have been implemented as three programs using Maple. In this section, we report the performance of the three programs, respectively.
All the timings in the tables are in seconds.

\begin{ex} \citep{Strzebonski}
\[f=ax^3+(a+b+c)x^2+(a^2+b^2+c^2)x+a^3+b^3+c^3-1.\]\
Under the order $a\prec b\prec c\prec x$, an open CAD defined by $f$ has $132$ sample points, while an open sample obtained by the algorithm \TwoHp\ has $15$ sample points.
\end{ex}
 \begin{ex}\citep{han2014constructing}
$$f=x^4-2x^2y^2+2x^2z^2+y^4-2y^2z^2+z^4+2x^2+2y^2-4z^2-4.$$ 
Under the order $z\succ y \succ x$, an open CAD defined by $f$ has $113$ sample points, while an open sample obtained by the algorithm \TwoHp\ has $87$ sample points.
\end{ex}
\begin{ex}
For $100$ random polynomials $f(x,y,z)$ with degree $8$\footnote{Generated by {\tt randpoly([x,y,z],degree=8)} in Maple 15.}, Figure \ref{fig:number} shows the numbers of real roots of $\Bproj(f,[z,y])$, $\Bproj(f,[y,z])$ and $\Hproj(f,[y,z])$, respectively. It is clear that the number of real roots of $\Hproj(f,[y,z])$ is always less than those of $\Bproj(f,[z,y])$ and $\Bproj(f,[y,z])$.
\begin{figure}[ht!]
        \begin{centering}
                \includegraphics[width=0.6\textwidth]{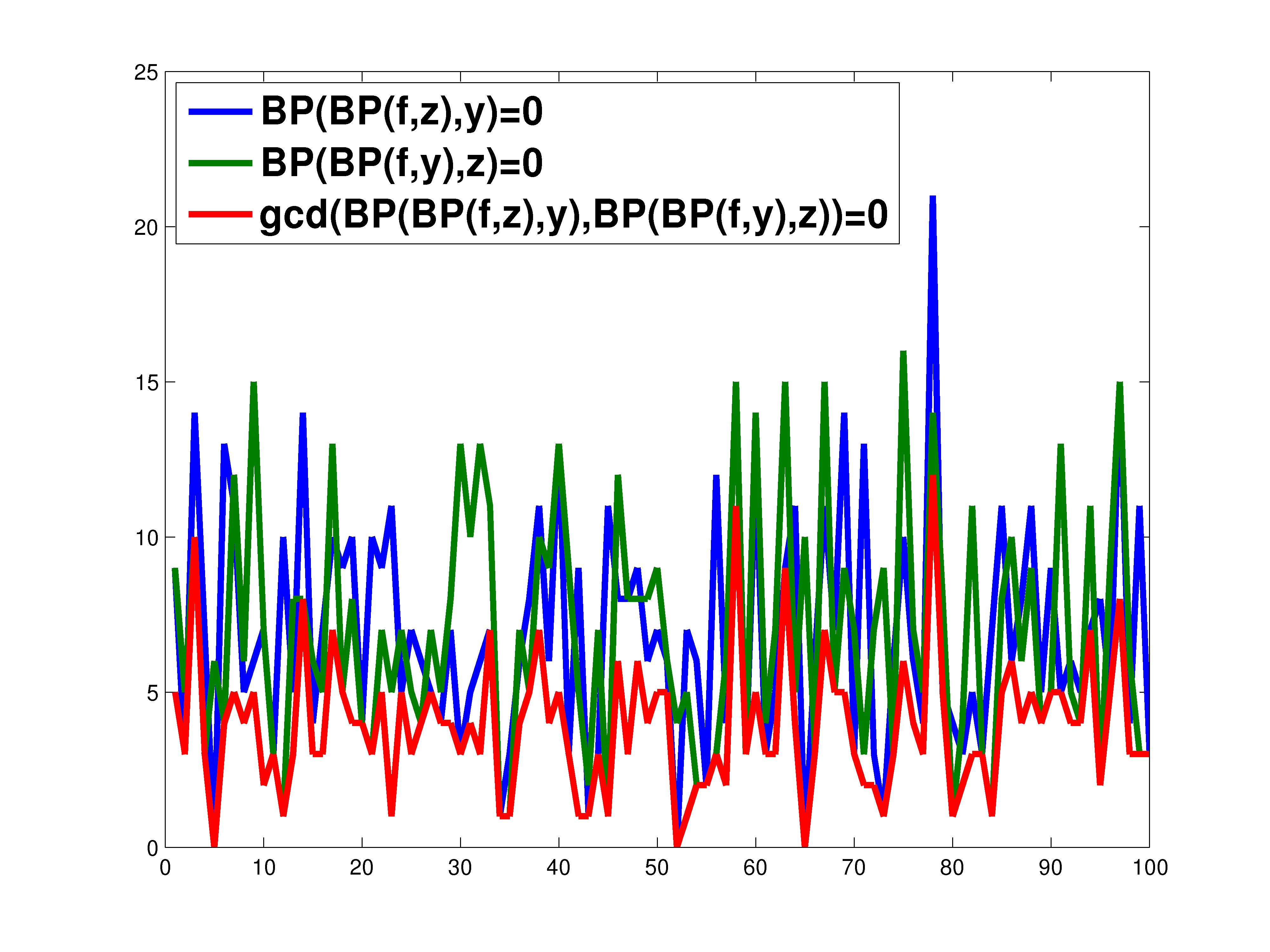}
                \caption{The number of real roots.} 
                \label{fig:number}
        \end{centering}
\end{figure}
\end{ex}

\begin{ex}\label{ex:61}
In this example, we compare the performance of Algorithm \TwoHp\ with open CAD on randomly generated polynomials. All the data in this example were obtained on a PC with Intel(R) Core(TM) i5 3.20GHz CPU, 8GB RAM, Windows 7 and Maple 17.

In the following table, we list the average time of projection phase and lifting phase, and the average number of sample points on $30$ random polynomials with $4$ variables and degree $4$ generated by {\tt randpoly([x,y,z,w],degree=4)-1}.

\begin{center}
                \begin{tabular}{ccccc}
                        \hline & {\rm Projection} &{\rm Lifting}& {\rm Sample\ points}  & \\
                        \hline
                        $\TwoHp$           &$0.13$& $ 0.29$& $262$&\\
                        ${\tt open\ CAD}$         & $0.19$&$ 3.11$&$ 486$&\\
                        \hline
                \end{tabular}
        \end{center}
If we get random polynomials with $5$ variables and degree $3$ by the command
$${\tt randpoly([seq(x[i],i=1..5)], degree=3)},$$ then the degrees of some variables are usually one. That makes the computation very easy for both \TwoHp\ and open CAD. Therefore,
we run the command ${\tt randpoly([seq(x[i],}$ ${\tt i=1..5)], degree=3)+add(x[i]^2,i=1..5)-1}$ 
 ten times to generate $10$ random polynomials with $5$ variables and degree $3$. 
The data on the $10$ polynomials are listed in the following table.
\begin{center}
                \begin{tabular}{ccccc}
                        \hline & {\rm Projection} &{\rm Lifting}& {\rm Sample\ points}  & \\
                        \hline
                        $\TwoHp$           &$2.87$& $ 3.51$& $2894$&\\
                        ${\tt open\ CAD}$         & $0.76$&$ 12.01$&$7802$&\\
                        \hline
                \end{tabular}
        \end{center}
For many random polynomials with $4$ variables and degree greater than $4$ (or $5$ variables and degree greater than $3$), neither \TwoHp\ nor open CAD can finish computation in reasonable time. 
\end{ex}

A main application of the new projection operator \Hproj\ is testing semi-definiteness of polynomials. 
Now, we illustrate the performance of our implementation of Algorithm \TwoPro\ and Algorithm \TCPT\ with several non-trivial examples. For more examples, please visit the homepage\footnote{\url{https://sites.google.com/site/jingjunhan/home/software}} of the first author.

We report the timings of the programs \TCPT, \TwoPro, and \Proineq\ \citep{han2016proving}, the function PartialCylindricalAlgebraicDecomposition (\PCAD) in Maple 15, function FindInstance (\FI) in Mathematica 9, QEPCAD B (\QEPCAD), the program \RAGlib \footnote{\RAGlib\ release 3.23 (Mar., 2015). The \RAGlib\ has gone through significant improvements. Thus, we updated the timing using the most recent version.}, and {\tt SOSTOOLS} in MATLAB \footnote{The MATLAB version is R2011b, SOSTOOLS's version is 3.00 and SeDuMi's version is 1.3.} on these examples.

\QEPCAD\ and {\tt SOSTOOLS} were performed on a PC with Intel(R) Core(TM) i5 3.20GHz CPU, 4GB RAM and ubuntu.
The other computations were performed on a laptop with Inter Core(TM) i5-3317U 1.70GHz CPU, 4GB RAM, Windows 8 and Maple 15.
\begin{ex}\label{ex:62} \citep{han2011} Prove that
        $$F(\bm{x}_{n})=(\sum_{i=1}^nx_i^2)^2-4\sum_{i=1}^n x_i^2x_{i+1}^2\ge 0,$$
        where $x_{n+1}=x_1$.

Hereafter ``$\infty$" means either the running time is over 4000 seconds or the software fails to get an answer.

        \begin{center}
                \begin{tabular}{lllllll}
                        \hline$n$ & 5 & 8& 11  & 17  &23  \\
                        \hline
            $\TCPT$          &0.06 &0.48&1.28&4.87 &11.95  \\
                        $\TwoPro$        &0.28 &0.95&6.26&29.53&140.01 \\
                        $\RAGlib$        &0.42 &0.76& $1.34$&$3.95$&$8.25$\\
                        $\Proineq$       &0.29 &$\infty$&$\infty$&$\infty$&$\infty$\\
            $\FI$            &0.10 &$\infty$ &$\infty$&$\infty$&$\infty$\\
            $\PCAD$          &0.26 &$\infty$ &$\infty$&$\infty$&$\infty$\\
            $\QEPCAD$        &0.10 &$\infty$ &$\infty$&$\infty$&$\infty$\\
            {\tt SOSTOOLS}   &0.23 & 1.38    & 3.94  & 247.56  & $\infty$\\
                        \hline
                \end{tabular}
        \end{center}

We then test the semi-definiteness of the polynomials (in fact, all $G(\bm{x}_{n})$ are indefinite)
$$G(\bm{x}_{n})=F(\bm{x}_{n})-\frac{1}{10^{10}}x_n^4.$$
The timings are reported in the following table.

        \begin{center}
                \begin{tabular}{lllllllll}
                        \hline$n$ &\TCPT& \TwoPro & \RAGlib&   \Proineq &  \FI& \PCAD&\QEPCAD\\
                        \hline
                        $20$       &1.81 &3.828&0.59& $\infty$& $\infty$&$\infty$&$\infty$\\
                        $30$       &5.59&13.594&2.01& $\infty$& $\infty$&$\infty$&$\infty$\\
           \hline
                \end{tabular}
        \end{center}
\end{ex}

\begin{ex} Prove that
         $$B(\bm{x}_{3m+2})=(\sum_{i=1}^{3m+2}x_i^2)^2-2\sum_{i=1}^{3m+2}x_i^2\sum_{j=1}^mx_{i+3j+1}^2\ge 0,$$
        where $x_{3m+2+r}=x_r$.
        If $m=1$, it is equivalent to the case $n=5$ of Example \ref{ex:62}. This form was once studied in \cite{parrilo2000structured}.
\end{ex}
        \begin{center}
                \begin{tabular}{lllllllll}
                        \hline$m$ &\TCPT& \TwoPro& \RAGlib&   \Proineq   &  \FI& \PCAD&\QEPCAD                \\
                        \hline
                        $1$       &0.03 &0.296&0.42&0.297&0.1&0.26&0.104 \\
                        $2$       &0.56 &1.390&0.36&23.094&$\infty$&$\infty$&$\infty$\\
                        $3$       &0.71 &9.672&0.75&$\infty$&$\infty$&$\infty$&$\infty$\\
            $4$       &7.68 &$\infty$&0.87&$\infty$&$\infty$&$\infty$&$\infty$\\
                        \hline
                \end{tabular}
        \end{center}

\begin{rem}

As showed by Example \ref{ex:61}, according to our experiments, the application of \TwoHp\ and \TwoPro\ is limited at $3$-$4$ variables and low degrees generally.
It is not difficult to see that, if the input polynomial $f(\xx_n)$ is symmetric, the new projection operator \Hproj\ cannot reduce the projection scale and the number of sample points.
Thus, it is reasonable to conclude that the complexity of \TwoPro\ is still doubly exponential.
\end{rem}

\section*{Acknowledgements}
Part of this paper is written while the first author visited Princeton University and North Carolina State University, and he would like to thank the institutes for their hospitality. He would like to thank ``Training, Research and Motion" (TRAM) network for supporting him visiting Princeton University. The first author would also like to thank his advisor Gang Tian for his constant support and encouragement.

The authors  thank M. Safey El Din who provided us several examples and communicated with us on the usage of \RAGlib.
The authors would like to convey their gratitude to all the four referees of our ISSAC'2014 paper and two referees of this paper, who provided their valuable comments, advice and suggestion, which helped improve this paper greatly on not only the presentation but also the technical details.

\bibliographystyle{elsart-harv}

\end{document}